\documentclass[reqno, 10pt]{amsart}

\usepackage{extsizes}
\usepackage{blindtext}

\usepackage{verbatim}
\usepackage{amsmath}
\usepackage{enumerate}
\usepackage{amssymb}
\usepackage{color}
\usepackage{url}
\usepackage{pdfpages}

\usepackage[normalem]{ulem} 
\usepackage{float} 

\usepackage{supertabular} 

\usepackage{mathabx}
\usepackage{multirow} 

\usepackage{subfig} 

\allowdisplaybreaks

\renewcommand{\phi}{\varphi}
\newcommand\+{\;\lower\plusheight\hbox{$+$}\;}
\newcommand\lldots{\;\lower\plusheight\hbox{$\cdots$}\;}

\newcommand{\ud}{\mathrm{d}}

\newcommand{\vect}[1]{\boldsymbol{#1}}
 
\newtheorem{Theorem}{Theorem}[section]

\setbox0=\hbox{$+$}
\newdimen\plusheight
\plusheight=\ht 0 \setbox0=\hbox{$-$}
\newdimen\minusheight
\minusheight=\ht0

\setbox0=\hbox{$\cdots$}
\newdimen\cdotsheight
\cdotsheight=\plusheight

\begin{document}
\title[]{The Cuboidal Lattices and their Lattice sums}

\author{Antony Burrows}
\address{Centre for Theoretical Chemistry and Physics, The New Zealand Institute for Advanced Study (NZIAS), Massey University Albany, Private Bag 102904, Auckland 0745, New Zealand}

\author{Shaun Cooper}
\address{School of Natural and Computational Sciences \\
Massey University--Albany \\
Private Bag 102904, North Shore Mail Centre\\
Auckland, New Zealand.\\
E-mail: s.cooper@massey.ac.nz 
}

\author{Peter Schwerdtfeger}
\address{Centre for Theoretical Chemistry and Physics, The New Zealand Institute for Advanced Study (NZIAS), Massey University Albany, Private Bag 102904, Auckland 0745, New Zealand}



\begin{abstract}
Lattice sums of cuboidal lattices, which connect the face-centered with the mean-centered and the body-centered cubic lattices through parameter dependent lattice vectors, are evaluated by decomposing them into two separate lattice sums related to a scaled cubic lattice and a scaled Madelung constant. Using theta functions we were able to derive fast converging series in terms of Bessel functions. Analytical continuations of these lattice sums are discussed in detail. 
\end{abstract}
\maketitle
\allowdisplaybreaks
\numberwithin{equation}{section}

\section{Introduction}

Lattice sums have a long history in solid-state physics and discrete mathematics \cite{BorweinEtAl}. They connect lattices to observables
such as the equation of state for a bulk system using interacting potentials between the lattice points (atoms or molecules) 
in three-dimensional space \cite{furth1944,Gruneisen1912,Smits2021,Stillinger2001}. Most notable cases for such interactions are the 
Lennard-Jones \cite{Jones-1925} and the Coulomb potential, leading in the latter case, for example, to the famous Madelung 
constant derived as early as in 1918 \cite{madelung1918}. For such potentials the corresponding lattice sums become functions of 
quadratic forms $\vect{i}^{\top}G\vect{i}$ with $\vect{i}\in\mathbb{Z}^3$ \cite{paper1}, i.e. the expression $\vect{i}^{\top}G\vect{i}$ 
is the quadratic form associated with the lattice $\mathcal{L}$ (or simply, the associated quadratic form).

In the general case of a $n$-dimensional lattice ($\vect{i}\in\mathbb{Z}^n$), $G$ is a positive definite, 
real and symmetric $(n\times n)$ matrix called the Gram matrix of the lattice $\mathcal{L}$, defined by its basis (or lattice) 
vectors $\{\vect{b}_i\}$ through $G=BB^\top$. $B=(\vect{b}_1,...,\vect{b}_n)^\top$ is called 
the generator matrix ($B$ not necessarily positive definite). Lattice sums represent often conditionally 
convergent series \cite{borwein1998convergence}, and the theory of converting them into fast converging 
series has become an intense research field over the past 50 years \cite{BorweinEtAl}.

Concerning the Gram matrix $G$ or generator matrix we introduce a few important definitions required in this work \cite{ConwayBook}. Two generator matrices $B_1$ and $B_2$ are equivalent if $B_2=cUB_1\mathcal{O}$, $c$ a non-zero real number, $\mathcal{O}$ a real orthogonal matrix ($\mathcal{O}\mathcal{O}^\top=1$) with $\textrm{det}(\mathcal{O})=\pm 1$ describing rotation, reflection or rotoreflection of the lattice, and $U$ a matrix containing integers with $\textrm{det}U=1$ describing for example permutations of the basis vectors. Given two equivalent generator matrices $B_1$ and $B_2$, the corresponding Gram matrices are related by
$$
G_2 = B_2 B_2^\top = cUB_1\mathcal{O} \left(cUB_1\mathcal{O}\right)^\top = c^2UB_1\mathcal{O}\mathcal{O}^\top B_1^\top U^\top 
= c^2U G_1 U^\top.
$$
\noindent
The minimum distance $d_\textrm{min}$ in a lattice $\mathcal{L}$ is defined by 
$$
d_\textrm{min}(\mathcal{L})=\textrm{min}\{d(\vect{v}_1;\vect{v}_2) | \vect{v}_1, \vect{v}_2 \in \mathcal{L}; \vect{v}_1\ne \vect{v}_2\}
$$
where $d(\vect{v}_1;\vect{v}_2) = |\vect{v}_1-\vect{v}_2|$ is the Euclidean distance.
In terms of the Gram matrix this is equivalent to 
$$
d_\textrm{min}=\textrm{min}\{+\sqrt{\vect{i}^{\top}G\vect{i}} ~|~ \vect{i}\in \mathbb{Z}^3\backslash (0,0,0)^\top \}.
$$
The minimal norm is related to the minimum distance by $\mu=d_\textrm{min}^2$. Dividing $G$ by $\mu$ assures that $d_\textrm{min}=1$ used in most lattice sum applications \cite{paper1}. For dense sphere packings the radius of a sphere $\rho$ is simply $\rho=d_\textrm{min}/2$. The packing density $\Delta_L$ and the center density $\delta_L$ of a three-dimensional lattice are given by 
$$
\Delta_\mathcal{L}=\frac{4\pi}{3}\delta_\mathcal{L}=\frac{4\pi}{3}\frac{\rho^3}{\textrm{vol}(\mathcal{L})}=\frac{4\pi}{3}\frac{\rho^3}{|\textrm{det}(B)|}=\frac{4\pi}{3}\frac{\rho^3}{\sqrt{\textrm{det}(G)}}. 
$$
The kissing number for dense sphere packings is defined by 
$$\textrm{kiss}(\mathcal{L})=\#\{\vect{v}\in \mathcal{L} ~|~ |\vect{v}|=d_\textrm{min}(\mathcal{L})\}.
$$

In this work we discuss cuboidal lattices and their lattice sums. We first present the characteristics of cuboidal lattices $\mathcal{L}(A)$ dependent on a single parameter $A$. In what follows we decompose the corresponding lattice sum into two lattice sums, where one is related to a scaled cubic lattice and the other to a scaled Madelung constant. We evaluate these lattice sums in two ways using  theta functions. We discuss these lattice sums including their analytical continuations and provide a more complete analysis for the lattice sum difference between f.c.c. (face centred cubic) and h.c.p. (hexagonal close packing).


\section{The cuboidal lattices}

Following Conway and Sloane~\cite[Sec. 3]{duals} we consider the lattice generated by the vectors
$(\pm u, \pm v, 0)^\top\quad\text{and}\quad (0,\pm v, \pm v)^\top$, where $u$ and $v$ are non-zero real numbers.
To make it specific, take the basis vectors
\begin{equation}
\vect{b}_1 = (u,v,0)^\top, \quad
\vect{b}_2 = (u,0,v)^\top, \quad
\vect{b}_3 = (0,v,v)^\top.
\label{b1b2b3}
\end{equation}
Then the generator matrix $B$ is given by
$$
B = \begin{pmatrix}
u & v & 0 \\
u & 0 & v \\
0 & v & v 
\end{pmatrix}
$$
which has determinant $-2uv^2$.
The Gram matrix is
\begin{equation}
\label{GramG}
G = B\,B^\top =  \begin{pmatrix}
u^2+v^2 & u^2 & v^2 \\
u^2 & u^2+v^2 & v^2 \\
v^2 & v^2 & 2v^2 
\end{pmatrix} 
= v^2\begin{pmatrix}
1+A & A & 1 \\
A & 1+A & 1 \\
1 & 1 & 2
\end{pmatrix}
\end{equation}
where $A=u^2/v^2$ and $G$ is positive definite for $A>0$. Conway and Sloane~\cite{duals} use $\sigma = u/v$, so $A = \sigma^2$. 
The most important cases, in decreasing numerical order, are:
\begin{enumerate}
\item $A=1$: the face-centred cubic (f.c.c.) lattice;
\item $A=1/\sqrt2$: the mean centred-cuboidal (m.c.c.) lattice;
\item $A=1/2$: the body-centred cubic (b.c.c.) lattice;
\item $A=1/3$: the axial centred-cuboidal (a.c.c.) lattice.
\end{enumerate}
The f.c.c. and b.c.c. lattices are well known. The corresponding Gram matrices for the f.c.c. and b.c.c lattices
are identical to the ones shown in our previous paper on lattice sums~\cite{paper1}. The m.c.c. and a.c.c. lattices
occur in~\cite{duals} and~\cite{conway2007optimal}. The m.c.c. lattice is the densest isodual lattice in three-dimensional space.

The quadratic form associated with the lattice is
\begin{align}
g(i,j,k) &= (i,j,k) \, G (i,j,k)^\top \nonumber \\
&= (u^2+v^2)i^2 + (u^2+v^2)j^2+2v^2k^2 + 2u^2ij + 2v^2ik+2v^2jk \nonumber \\
&= u^2(i^2+j^2)+v^2(j+k)^2+v^2(i+k)^2 \nonumber \\
&=v^2\left(A(i+j)^2+(j+k)^2+(i+k)^2\right). \label{Gg}
\end{align}
It is easy to check that
\begin{equation}
\label{minnorm}
\min_{i,j,k\in\mathbb{Z} \atop (i,j,k) \neq (0,0,0)}\{A(i+j)^2+(j+k)^2+(i+k)^2 \}
= \begin{cases}
4A & \text{if $0<A<1/3$,} \\
A+1 & \text{if $1/3\leq A \leq 1$,} \\
2 & \text{if $A>1$.}
\end{cases} 
\end{equation}
It follows from~\eqref{Gg} and~\eqref{minnorm} that the minimum distance is given by
\begin{equation}
\label{pd1}
d_\textrm{min} 
= \min_{i,j,k\in\mathbb{Z} \atop (i,j,k) \neq (0,0,0)} \, \sqrt{g(i,j,k)} = \begin{cases}
2v\sqrt{A} & \text{if $0<A<1/3$,} \\
v\sqrt{A+1} & \text{if $1/3\leq A \leq 1$,} \\
v\sqrt2 & \text{if $A>1$.}
\end{cases} 
\end{equation}
We rescale to make the minimum distance~$1$ by defining
\begin{align*}
g(A;i,j,k) &= \frac{g(i,j,k)}{\left(d_\textrm{min} \right)^2} \\
&= \begin{cases}
\displaystyle{\frac{1}{4A}\left(A(i+j)^2+(j+k)^2+(i+k)^2\right)} & \text{if $0<A<1/3$,} \\ \\
\displaystyle{\frac{1}{A+1}\left(A(i+j)^2+(j+k)^2+(i+k)^2\right)} & \text{if $1/3\leq A \leq 1$,} \\ \\
\displaystyle{\frac{1}{2}\left(A(i+j)^2+(j+k)^2+(i+k)^2\right)} & \text{if $A>1$.}
\end{cases} 
\end{align*}
The examples we are interested are in (f.c.c., m.c.c., b.c.c., a.c.c.) all satisfy \mbox{$1/3\leq A \leq 1$}.
Because of its practical interest, this is the only case we will study, and from here on (unless otherwise mentioned)
we will always assume \mbox{$1/3\leq A \leq 1$}
in which case we have
\begin{equation}
\label{gdef}
g(A;i,j,k) = \frac{1}{A+1}\left(A(i+j)^2+(j+k)^2+(i+k)^2\right),
\end{equation}
corresponding to the rescaled Gram matrix
\begin{equation}
G(A):= \frac{1}{\left(d_\textrm{min} \right)^2}\;G = \frac{1}{A+1}\begin{pmatrix}
1+A & A & 1 \\
A & 1+A & 1 \\
1 & 1 & 2
\end{pmatrix}.
\end{equation}
The packing density is calculated as
$$
\Delta_\mathcal{L}=\frac{4\pi}{3}\frac{\rho^3}{\sqrt{\textrm{det}(G)}}
$$
where $\rho = d_\textrm{min}/2$ and
$$
\det(G) = \left(\det(B)\right)^2 = 4u^2v^4 = 4Av^6.
$$
It follows that
$$
\Delta_\mathcal{L}=\frac{\pi}{12 \, \sqrt{A}} \left(\frac{d_\textrm{min}}{v}\right)^3.
$$
On using the formula for $d_\textrm{min}$ in~\eqref{pd1} we obtain the formula for the packing density, given by
$$
\Delta_\mathcal{L}
= \begin{cases}
\displaystyle{\frac{2\pi A}{3}} & \text{if $0<A<1/3$,} \\ \\
\displaystyle{\frac{\pi}{12}\sqrt{\frac{(A+1)^3}{A}}} & \text{if $1/3\leq A \leq 1$,} \\ \\
\displaystyle{\frac{\pi}{6}\sqrt{\frac{2}{A}}} & \text{if $A>1$.}
\end{cases} 
$$
Figure~1 shows a graph of the packing density as a function of the parameter $A$.
Further information is recorded in Table~\ref{tab:lattice}.
In the main region of interest $1/3\leq A \leq 1$, we have
\begin{equation}
\label{pdA}
\Delta_\mathcal{L} =\frac{\pi}{12}\sqrt{\frac{(A+1)^3}{A}}
\end{equation}
and so
\begin{equation}
\label{pdA1}
\frac{\ud \Delta_\mathcal{L} }{\ud A} =\frac{\pi}{12}\left(A-\frac12\right)\sqrt{\frac{A+1}{A^3}}.
\end{equation}
It follows that on the interval $1/3\leq A \leq 1$, the packing density has a maximum of $\pi\sqrt{2}/6 \approx 0.74048$ at $A=1$
corresponding to f.c.c., and a minimum of \mbox{$\pi\sqrt{3}/8 \approx 0.68017$} at $A=1/2$ corresponding to b.c.c.

It is also interesting to note that as  $A\rightarrow\infty$ the limiting case of the lattice is the two-dimensional square close 
packing with minimal distance~$1$ and kissing number~$4$; while in the other extreme case the limit as $A\rightarrow 0$
gives the one-dimensional lattice with minimal distance~$1$ and kissing number~$2$. These cases are briefly analysed in Appendix~B.

\begin{table}[htp]
\setlength{\tabcolsep}{0.1cm}
\caption{\label{tab:lattice} Kissing number, packing density $\Delta_\mathcal{L}$ and integer combinations $\vect{i}_n$ for the lattice
associated with the Gram matrix $G$ defined in \eqref{GramG}. The values in the table depend only on $A$
and are independent of $v$.}
\begin{center}
\begin{tabular}{l|l|l|l|l}
\hline
Region & $A$ & $\textrm{kiss}(\mathcal{L})$ & $\Delta_\mathcal{L}$ & integer combinations$^a$\\
\hline
I 		& $(0,\frac{1}{3})$	& 2 		& $\frac{2\pi A}{3}$	& $\vect{i}_1,\vect{i}_2$ \\
a.c.c. 	& $\frac{1}{3}$		& 10 & $\frac{2\pi}{9}$	& $\vect{i}_1,...,\vect{i}_{10}$\\
II 		& $(\frac{1}{3},1)$	& 8 	& $\frac{\pi}{12}\sqrt{\frac{(A+1)^3}{A}}$	&  $\vect{i}_3,...,\vect{i}_{10}$\\
f.c.c. 	& $1$			& 12 & $\frac{\pi\,\sqrt{2}}{6}$	&  $\vect{i}_3,...,\vect{i}_{14}$ \\
III 		& $(1,\infty)$		& 4 	& $\frac{\pi}{6}\sqrt{\frac{2}{A}}$	&  $\vect{i}_{11},...,\vect{i}_{14}$\\
\hline
\end{tabular}
\end{center}
\scriptsize{$^a$ The integer combinations $\vect{i}$ which determine $d_\textrm{min}$ in~\eqref{pd1} for the different regions are as follows: $\vect{i}_1^\textrm{I}=(-1,-1,1)$, $\vect{i}_2^\textrm{I}=(1,1,-1)$, $\vect{i}_3^\textrm{II}=(-1,0,0)$, $\vect{i}_4^\textrm{II}=(-1,0,1)$, $\vect{i}_5^\textrm{II}=(0,-1,0)$, $\vect{i}_6^\textrm{II}=(0,-1,1)$, $\vect{i}_7^\textrm{II}=(0,1,-1)$, $\vect{i}_8^\textrm{II}=(0,1,0)$, $\vect{i}_9^\textrm{II}=(1,0,-1)$, $\vect{i}_{10}^\textrm{II}=(1,0,0)$, $\vect{i}_{11}^\textrm{III}=(-1,1,0)$, $\vect{i}_{12}^\textrm{III}=(0,0,-1)$, $\vect{i}_{13}^\textrm{III}=(0,0,1)$, $\vect{i}_{14}^\textrm{III}=(1,-1,0)$. }
\end{table}

\medskip

\begin{figure*}[htbp]
\begin{center}
\includegraphics[width=11cm]{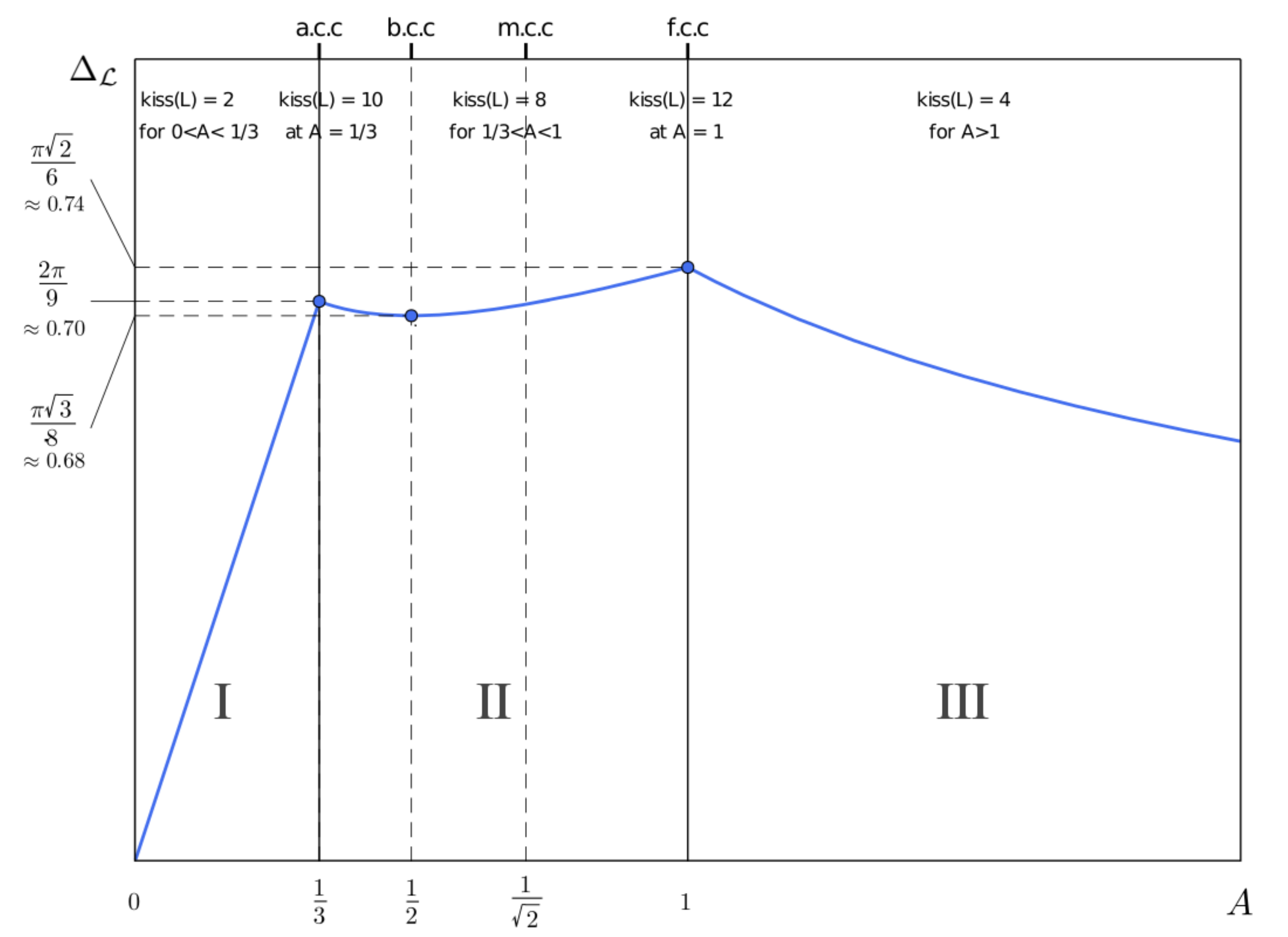}\\
\caption{
{\color{black}{Graph of the packing density $\Delta_\mathcal{L}$ versus~$A$.}} The regions I, II and III are according to the different kissing numbers.
Explicit formulas are given in Table~\ref{tab:lattice}. 
The location of the f.c.c., m.c.c., b.c.c. and a.c.c. lattices are indicated on the graph.}
\label{fig:ALattice}
\end{center}
\end{figure*}

Given a positive definite quadratic form $g(i,j,k)$, the corresponding theta series is defined for $|q|<1$ by
$$
\theta_g(q) = \sum_{i=-\infty}^\infty \sum_{j=-\infty}^\infty \sum_{k=-\infty}^\infty q^{g(i,j,k)}.
$$
For the quadratic form in~\eqref{gdef} the theta series is
$$
\theta(A;q) = \sum_{i=-\infty}^\infty \sum_{j=-\infty}^\infty \sum_{k=-\infty}^\infty  q^{(A(i+j)^2+(j+k)^2+(i+k)^2)/(A+1)} \quad \text{ where $1/3\leq A \leq 1$.}
$$
The first few terms in the theta series for f.c.c., m.c.c., b.c.c. and a.c.c. as far as~$q^9$ are given respectively by
\begin{align*}
\theta(1;q) &= 1+12q+6q^2+24q^3+12q^4+24q^5+8q^6+48q^7+6q^8+36q^9+\cdots, \\
\theta(\frac{1}{\sqrt{2}};q) &= 1+8q+4q^{4-2\sqrt2} + 2 q^{4 \sqrt2-4} +4 q^{8-4 \sqrt2}+8 q^{2 \sqrt2} + 16q^{-4\sqrt2+9} \\
&\quad +8q^4+8q^{8\sqrt2-7}+4 q^{16-8\sqrt2}+8 q^{-8\sqrt2+17} + 8q^{20-10\sqrt2} + 8q^{-4\sqrt2+12} \\
&\quad+ 2q^{16\sqrt2-16} + 16q^{4\sqrt2+1} + 16q^{-6\sqrt2+16} + 8q^{14\sqrt2-12}+16q^{-12\sqrt2+25} \\
&\quad+8q^{-8+12\sqrt2} + 8q^9 +\cdots, \\
\theta(\frac12;q) &= 1+8q+6q^{4/3}+12q^{8/3}+8q^4+24q^{11/3}+6q^{16/3}+24q^{19/3}+24q^{20/3} \\
&\quad +24q^8+32q^9+\cdots, \\
\theta(\frac13;q) &= 1 + 10q + 4q^{3/2}+8q^{5/2}+12q^3+26q^4+8q^{11/2}+20q^6+32q^7 \\
&\quad +8q^{15/2}+16q^{17/2}+10q^9+\cdots.
\end{align*}
Since the quadratic form $g(A;i,j,k)$ has been normalised to make the minimum distance~$1$, the kissing number occurs in each theta series as the coefficient of~$q$.
That is, we have $\textrm{kiss(f.c.c.)}=12$, $\textrm{kiss(m.c.c.)}=8$, $\textrm{kiss(b.c.c.)}=8$ and $\textrm{kiss(a.c.c.)}=10$.

We introduce the following lattice sum important in solid state theory \cite{BorweinEtAl},
\begin{equation}
\label{LAsdefinition}
L(A;s) = {\sum_{i,j,k}}^{\prime} \left(\frac{1}{g(A;i,j,k)}\right)^s = {\sum_{i,j,k}}^{\prime} \left(\frac{A+1}{A(i+j)^2+(j+k)^2+(i+k)^2}\right)^s
\end{equation}
where $1/3\leq A \leq 1$.
Here and throughout this work, a prime on the summation symbol  will denote
that the sum ranges over all integer values except for the term when all of the summation indices are simultaneously zero.
Thus, the sums in~\eqref{LAsdefinition} are over all integer values
 of $i$, $j$ and $k$ except for the term $(i,j,k)=(0,0,0)$ which is omitted.
This lattice sum smoothly connects four different well known lattices, i.e., when $A=1$, $1/\sqrt{2}$, $1/2$ or $1/3$ we obtain the expressions for the
lattices f.c.c, m.c.c., b.c.c. and a.c.c. respectively (face-centred cubic, mean centred-cuboidal, body-centred cubic, and axial centred cuboidal).
In these cases, we also write
\begin{align*}
L_3^{\text{FCC}}(s) &= L(1;s), \\
L_3^{\text{MCC}}(s) &= L(1/\sqrt{2};s), \\
L_3^{\text{BCC}}(s) &= L(1/2;s), \\
\text{and}\quad L_3^{\text{ACC}}(s) &= L(1/3;s).
\end{align*}

We conclude this section by reconciling the Gram matrix $G$ in~\eqref{GramG} with two 
matrices given by Conway and Sloane~\cite{duals}. Let
$$
U_1 = \begin{pmatrix}
1 & 0 & 0 \\
-1 & 0 & 1 \\
0 & -1 & 0 
\end{pmatrix}
\quad\text{and} \quad
U_2 = \begin{pmatrix}
1 & 1 & -1 \\
1 & 0 & 0 \\
0 & 1 & 0 
\end{pmatrix}
$$
and consider the equivalent matrices $G_1$ and $G_2$ defined by
\begin{equation}
\label{GramG1}
G_1 = U_1\,G\,U_1^\top =  \begin{pmatrix}
u^2+v^2 & -u^2 & -u^2 \\
-u^2 & u^2+v^2 & u^2-v^2 \\
-u^2 & u^2-v^2 & u^2+v^2 
\end{pmatrix}
\end{equation}
and
\begin{equation}
\label{GramG2}
G_2 = U_2\,G\,U_2^\top =  \begin{pmatrix}
4u^2 & 2u^2 & 2u^2 \\
2u^2 & u^2+v^2 & u^2 \\
2u^2 & u^2 & u^2+v^2 
\end{pmatrix}.
\end{equation}
When $u=1/\sqrt{2}$ and $v=1/\sqrt[4]{2}$, the matrix $G_1$ in~\eqref{GramG1} is the Gram matrix
for the m.c.c lattice given by Conway and Sloane~\cite[(10)]{duals}.
Moreover, when $u=\sqrt{1/3}$ and $v=\sqrt{2/3}$, the matrix $G_1$ leads to another well-known
quadratic form for the b.c.c. lattice, e.g., see~\cite[(8b)]{paper1}.
 When $u=1$, $v=\sqrt{3}$, the matrix $G_2$ in~\eqref{GramG2} is the Gram matrix for the a.c.c. lattice given
in~\cite[p. 378]{duals}.
Since $\det U_1 = \det U_2 = 1$ it follows that
$$
\det G_1 = \det G_2 = \det G = (\det B)^2 = 4u^2v^4 = 4v^6A.
$$
The corresponding quadratic forms $g_1$ and $g_2$ are defined by
\begin{align*}
g_1(i,j,k) &= (i,j,k) \, G_1 (i,j,k)^\top \\
&= (u^2+v^2)i^2 + (u^2+v^2)j^2+(u^2+v^2)k^2 - 2u^2ij - 2u^2ik+2(u^2-v^2)jk 
\intertext{and}
g_2(i,j,k) &= (i,j,k) \, G_2 (i,j,k)^\top \\
&= 4u^2i^2 + (u^2+v^2)j^2+(u^2+v^2)k^2 + 4u^2ij + 4u^2ik+2u^2jk.
\end{align*}
They are related to the quadratic form $g$ in~\eqref{Gg} by
\begin{equation}
\label{g1qf}
g_1(i,j,k) = g\left( (i,j,k)U_1 \right) = g(i-j,-k,j) 
\end{equation}
and
$$
g_2(i,j,k) = g\left( (i,j,k)U_2 \right) = g(i+j,i+k,-i).
$$
\\


\section{A minimum property of the lattice sum $L(A;s)$}

In the previous section --- see~\eqref{pdA} and~\eqref{pdA1} --- it was noted that on the interval $1/3 \leq A \leq 1$, the packing density function~$\Delta_\mathcal{L}$
has a minimum value when~$A=1/2$. The next result shows that provided $s>3/2$, the corresponding lattice sum $L(A;s)$
also attains a minimum at the same value $A=1/2$.

\begin{Theorem}
\label{theorem3point1}
Let $L(A;s)$ be the lattice defined by~\eqref{LAsdefinition}, that is,
$$
L(A;s) = {\sum_{i,j,k}}^{\prime} \left(\frac{1}{g(A;i,j,k)}\right)^s = {\sum_{i,j,k}}^{\prime} \left(\frac{A+1}{A(i+j)^2+(j+k)^2+(i+k)^2}\right)^s
$$
where $s>3/2$ and $1/3\leq A \leq 1$. 
Then 
$$
\left.\frac{\partial}{\partial A} L(A;s) \right|_{A=1/2} =0
\quad \text{and}\quad \left.\frac{\partial^2}{\partial A^2} L(A;s) \right|_{A=1/2} >0.
$$
\end{Theorem}
\begin{proof}
By definition we have
$$
L(A;s) = {\sum_{I,J,K}}^{\prime} \left( \frac{1}{g(A;I,J,K)}\right)^s
$$
where
$$
g(A;I,J,K) = \frac{1}{A+1}\left(A(I+J)^2+(J+K)^2+(I+K)^2\right).
$$
Now make the change of variables given by~\eqref{g1qf}, namely
$$
(I,J,K) = (i-j,-k,j).
$$
This is a bijection since
$$
(i,j,k) = (I+K,K,-J),
$$
and it follows that
\begin{align*}
L(A;s) 
&= {\sum_{i,j,k}}^{\prime} \left( \frac{1}{g(A;i-j,-k,j)}\right)^s \\
&={\sum_{i,j,k}}^{\prime} \frac{1}{ \left(i^2+j^2+k^2-2(ij+ik)\left(\frac{A}{A+1}\right)+2jk\left(\frac{A-1}{A+1}\right)\right)^s}.
\end{align*}
By direct calculation, the derivative is given by
\begin{equation}
\label{deriv0}
\frac{\partial}{\partial A} L(A;s) =\frac{2s}{(A+1)^2}
{\sum_{i,j,k}}^{\prime}  \frac{ij+ik-2jk}{ \left(i^2+j^2+k^2-2(ij+ik)\left(\frac{A}{A+1}\right)+2jk\left(\frac{A-1}{A+1}\right)\right)^{s+1}}.
\end{equation}
Setting $A=1/2$ gives
\begin{equation}
\label{deriv1}
\left.\frac{\partial}{\partial A} L(A;s)\right|_{A=1/2}
 =\frac{8s}{9}\;
{\sum_{i,j,k}}^{\prime}  \frac{ij+ik-2jk}{ \left(i^2+j^2+k^2-\frac23(ij+ik+jk)\right)^{s+1}}.
\end{equation}
Switching $i$ and $j$ gives
\begin{equation}
\label{deriv2}
\left.\frac{\partial}{\partial A} L(A;s)\right|_{A=1/2}
 =\frac{8s}{9}\;
{\sum_{i,j,k}}^{\prime}  \frac{ij+jk-2ik}{ \left(i^2+j^2+k^2-\frac23(ij+ik+jk)\right)^{s+1}},
\end{equation}
while switching $i$ and $k$ in~\eqref{deriv1} gives
\begin{equation}
\label{deriv3}
\left.\frac{\partial}{\partial A} L(A;s)\right|_{A=1/2}
 =\frac{8s}{9}\;
{\sum_{i,j,k}}^{\prime}  \frac{jk+ik-2ij}{ \left(i^2+j^2+k^2-\frac23(ij+ik+jk)\right)^{s+1}}.
\end{equation}
On adding \eqref{deriv1}, \eqref{deriv2} and~\eqref{deriv3} and noting that
$$
(ij+ik-2jk) + (ij+jk-2ik) + (jk+ik-2ij) = 0
$$
it follows that
$$
\left.\frac{\partial}{\partial A} L(A;s)\right|_{A=1/2} = 0.
$$
Next, taking the derivative of~\eqref{deriv0} gives
\begin{align*}
\frac{\partial^2}{\partial A^2} L(A;s) 
&=
\frac{-4s}{(A+1)^3}
{\sum_{i,j,k}}^{\prime}  \frac{ij+ik-2jk}{ \left(i^2+j^2+k^2-2(ij+ik)\left(\frac{A}{A+1}\right)+2jk\left(\frac{A-1}{A+1}\right)\right)^{s+1}} \\
&\quad + \frac{4s(s+1)}{(A+1)^4}
{\sum_{i,j,k}}^{\prime}  \frac{(ij+ik-2jk)^2}{ \left(i^2+j^2+k^2-2(ij+ik)\left(\frac{A}{A+1}\right)+2jk\left(\frac{A-1}{A+1}\right)\right)^{s+2}}.
\end{align*}
When $A=1/2$ the first sum is zero by the calculations in the first part of the proof. Therefore,
\begin{align*}
\left.\frac{\partial^2}{\partial A^2} L(A;s) \right|_{A=1/2}
&= \frac{64s(s+1)}{81}
{\sum_{i,j,k}}^{\prime}  \frac{(jk+ik-2ij)^2}{ \left(i^2+j^2+k^2-\frac23(ij+ik+jk)\right)^{s+2}}.
\end{align*}
The term $(jk+ik-2ij)^2$ in the numerator is non-negative and not always zero. The
denominator is always positive because the quadratic form is positive definite. It follows that 
$$
\left.\frac{\partial^2}{\partial A^2} L(A;s) \right|_{A=1/2} >0
$$
as required.

The calculations above are valid provided term-by-term differentiation of the series is allowed. All of the series encountered above
converge absolutely and uniformly on compact subsets of the region Re$(s)>3/2$. On restricting~$s$ to real values, the conclusion
about positivity is valid for $s>3/2$.
\end{proof}

A consequence of Theorem~\ref{theorem3point1} is that for any fixed value $s>3/2$, the lattice
sum $L(A;s)$ attains a minimum when~\mbox{$A=1/2$.} 
Graphs of $y=L(A;s)$ to illustrate this minimum property are shown in Fig.~2.
In the limiting case $s\rightarrow\infty$ we have
$$
L(A;\infty) = \lim_{s\rightarrow\infty}L(A;s) = \textrm{kiss}(\mathcal{L}) 
=  \begin{cases}
10 & \text{if $A=1/3$,} \\
8 & \text{if $1/3 < A < 1$,} \\
12 & \text{if $A=1$.}
\end{cases} 
$$
This graph is also shown in Fig.~2.

 \begin{figure*}[htbp]
  \begin{center}
  \includegraphics[scale=0.33]{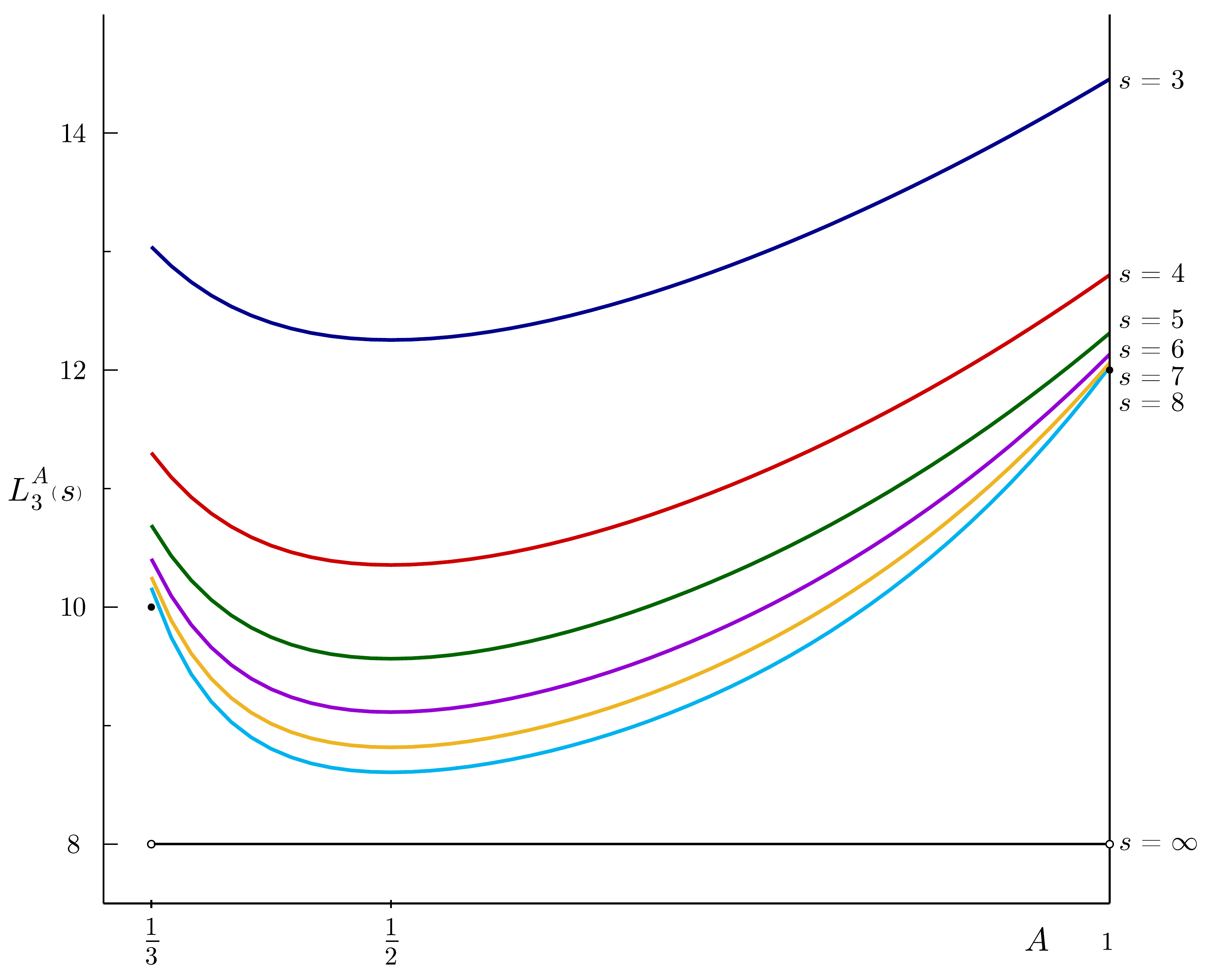}\\
  \caption{
  {\color{black}{Graph of $L(A;s)$ versus $A$ for various values of $s$.}}
  }
  \label{fig:LatticeSums}
  \end{center}
  \end{figure*}


\section{Evaluation of the sum $L(A;s)$}
We now turn to the evaluation of $L(A;s)$. Our objectives are to find formulas that are both simple
and computationally efficient. The formulas we obtain can be used to show that $L(A;s)$ can be analytically continued
to complex values of~$s$, with a simple pole at $s=3/2$ and no other singularities.

One method of evaluating the sum $L(A;s)$ is to use the Terras decomposition. This was done for 
f.c.c. and b.c.c. in~\cite{paper1} and can in principle also be used for $L(A;s)$. 
Here we use an easier method that also works the whole parameter range $1/3 \leq A \leq 1$
and hence gives the lattice sum for all four lattices f.c.c., m.c.c., b.c.c. and a.c.c. In fact, the advantage here is that 
we obtain two formulas which not only can be used as checks, but also contain different information about their analytic continuation.

We begin by writing the lattice sum in the form
\begin{align}
L(A;s) &= {\sum_{i,j,k}}^{\prime} \left(\frac{A+1}{A(i+j)^2+(j+k)^2+(i+k)^2}\right)^s \nonumber \\
&={\sum_{I,J,K \atop I+J+K\, \text{even}}}^{\!\!\!\!\!\!\!\!\prime} \left(\frac{A+1}{AI^2+J^2+K^2}\right)^s \nonumber \\
&= \frac{(A+1)^s}{2} {\sum_{i,j,k}}^{\prime} \;\;\; \frac{1+(-1)^{i+j+k}}{(Ai^2+j^2+k^2)^s}. \label{easier}
\end{align}
Therefore, we evaluate the sums 
\begin{equation}
\label{sum1}
T_1(A;A;s) := {\sum_{i,j,k}}^{\prime} \;\;\; \frac{1}{(Ai^2+j^2+k^2)^s}
\end{equation}
and
\begin{equation}
\label{sum2}
T_2(A;s) := {\sum_{i,j,k}}^{\prime} \;\;\; \frac{(-1)^{i+j+k}}{(Ai^2+j^2+k^2)^s}.
\end{equation}
By \eqref{easier}, \eqref{sum1} and \eqref{sum2}, the required lattice sum is given by
\begin{equation}
\label{sum3}
L(A;s) =  \frac{(A+1)^s}{2} \left(T_1(A;s)  + T_2(A;s)  \right).
\end{equation}

\subsection{The sum $T_1(A;s)$}
We shall consider two ways for handling the sum in~\eqref{sum1}. The first is to separate the terms according to whether $i=0$ or $i\neq 0$,
which gives rise to
\begin{equation}
\label{d1}
T_1(A;s) = f(s) + 2F(s)
\end{equation}
where
\begin{align*}
f(s) &= {\sum_{j,k}}^{\,\prime} \;\;\; \frac{1}{(j^2+k^2)^s}
\intertext{and}
F(s) &= \sum_{i=1}^\infty \sum_{j=-\infty}^\infty \sum_{k=-\infty}^\infty \frac{1}{(Ai^2+j^2+k^2)^s}.
\end{align*}
This is the starting point of the approach taken by Selberg and Chowla~\cite[Section~7]{selberg}.
Another way is to separate the terms
according to whether $(j,k)=(0,0)$ or $(j,k) \neq (0,0)$ and write
\begin{equation}
\label{d2}
T_1(A;s) = 2g(s) + G(s)
\end{equation}
where
\begin{align*}
g(s)
&= \sum_{i=1}^\infty \frac{1}{(Ai^2)^s} 
\intertext{and}
G(s) &=
{\sum_{j,k}}^{\,\prime} \; \sum_{i=-\infty}^\infty \frac{1}{(Ai^2+j^2+k^2)^s}.
\end{align*}
The series $F(s)$, $g(s)$ and $G(s)$ also depend on~$A$.
For simplicity we omit the parameter~$A$ from the notation and just write $F(s)$, $g(s)$ and $G(s)$ in place of $F(A;s)$, $g(A;s)$ and $G(A;s)$, respectively.

We will now analyse~\eqref{d1}; the corresponding analysis for~\eqref{d2} will be given in Section~\ref{sec4point2}.
By the well-known result~\eqref{zeta1} we have
$$
f(s)= {\sum_{j,k}}^{\,\prime} \;\;\; \frac{1}{(j^2+k^2)^s} = 4\zeta(s)L_{-4}(s)
$$
where
$$
\zeta(s) = \sum_{n=1}^\infty \frac{1}{n^s}
$$
is the Riemann zeta function, and
$$
L_{-4}(s) = \sum_{n=1}^\infty \frac{\sin \frac{n\pi}{2}}{n^s} = \frac{1}{1^s}-\frac{1}{3^s}+\frac{1}{5^s}
-\frac{1}{7^s}+\cdots.
$$
is the Dirichlet beta series. It remains to analyse $F(s)$. By the integral formula for the gamma function~\eqref{gamma2} we have
\begin{align*}
\pi^{-s}\Gamma(s)F(s)
&= \int_0^\infty x^{s-1} \sum_{i=1}^\infty e^{-\pi A x i^2} 
\sum_{j=-\infty}^\infty \sum_{k=-\infty}^\infty e^{-\pi x(j^2+k^2)} \; \ud x \\
&= \int_0^\infty x^{s-1} \sum_{i=1}^\infty e^{-\pi A x i^2} 
\left(\sum_{j=-\infty}^\infty e^{-\pi xj^2}\right)^2 \; \ud x.
\intertext{Now apply the modular transformation for theta functions~\eqref{theta2} to obtain}
\pi^{-s}\Gamma(s)F(s)
&= \int_0^\infty x^{s-1} \sum_{i=1}^\infty e^{-\pi A x i^2} 
\left(\frac{1}{\sqrt{x}}\sum_{j=-\infty}^\infty e^{-\pi j^2/x}\right)^2 \; \ud x \\
&= \int_0^\infty x^{s-2} \sum_{i=1}^\infty e^{-\pi A x i^2} 
\sum_{N=0}^\infty r_2(N) e^{-\pi N/x}\ \; \ud x
\intertext{where $r_2(N)$ is the number of representations of $N$ as a sum of two squares, e.g., see~\eqref{theta3}. Now separate out the $N=0$ term and
evaluate the resulting integrals. We find that}
\pi^{-s}\Gamma(s)F(s)
&= \sum_{i=1}^\infty\int_0^\infty x^{s-2} e^{-\pi A x i^2}  \; \ud x
+\sum_{i=1}^\infty\sum_{N=1}^\infty  r_2(N) 
\int_0^\infty x^{s-2}  e^{-\pi A x i^2-\pi N/x}\ \; \ud x \\
&= \frac{\Gamma(s-1)\zeta(2s-2)}{A^{s-1}\pi^{s-1}}
+ 2\sum_{i=1}^\infty\sum_{N=1}^\infty  r_2(N) \left(\frac{N}{Ai^2}\right)^{(s-1)/2} K_{s-1}\left(2\pi i \sqrt{AN}\right)
\end{align*}
where we have used the formula~\eqref{bessel1} for the $K$-Bessel function.
On using all of the above back in~\eqref{d1} we obtain
\begin{align}
&{\sum_{i,j,k}}^{\prime} \;\;\; \frac{1}{(Ai^2+j^2+k^2)^s}
= 4\zeta(s)L_{-4}(s) +\frac{2\pi}{(s-1)}\frac{ \zeta(2s-2)}{A^{s-1}} \nonumber \\
&\quad +\frac{4\pi^s}{\Gamma(s)}\, A^{(1-s)/2}\, \sum_{i=1}^\infty\sum_{N=1}^\infty  r_2(N) \left(\frac{N}{i^2}\right)^{(s-1)/2} K_{s-1}\left(2\pi i \sqrt{AN}\right).
\label{sumpart1}
\end{align}
This is essentially Selberg and Chowla's formula~\cite[(45)]{selberg}.
They write it in a slightly different form in terms of a sum over the divisors of $N$ to minimise the number of
Bessel function evaluations. We will leave it as it is for simplicity.

\subsection{A second formula for the sum $T_1(A;s)$}
\label{sec4point2}

This time we split the terms according as in~\eqref{d2} and start with
\begin{equation}
\label{d3}
T_1(A;s) = 2g(s) + G(s)
\end{equation}
where
\begin{align*}
g(s)
&= \sum_{i=1}^\infty \frac{1}{(Ai^2)^s} = A^{-s} \,\zeta(2s)
\intertext{and}
G(s) &=
{\sum_{j,k}}^{\,\prime} \; \sum_{i=-\infty}^\infty \frac{1}{(Ai^2+j^2+k^2)^s}.
\end{align*}
Now apply the integral formula for the gamma
function~\eqref{gamma2} and then the modular transformation for the theta function~\eqref{theta1} to obtain
\begin{align*}
\pi^{-s}\Gamma(s)G(s)
&= \int_0^\infty x^{s-1} {\sum_{j,k}}^{\;\prime} e^{-\pi(j^2+k^2)x}
\sum_{i=-\infty}^\infty e^{-\pi i^2 Ax} \, \ud x \\
&= \frac{1}{\sqrt{A}}\int_0^\infty x^{s-3/2}\, {\sum_{j,k}}^{\;\prime} e^{-\pi(j^2+k^2)x}
\sum_{i=-\infty}^\infty e^{-\pi i^2/Ax} \, \ud x.
\intertext{Separate the $i=0$ term, to get}
\pi^{-s}\Gamma(s)G(s)
&=\frac{1}{\sqrt{A}} \int_0^\infty x^{s-3/2}\, {\sum_{j,k}}^{\;\prime} e^{-\pi(j^2+k^2)x}
\, \ud x \\
& \quad + \frac{2}{\sqrt{A}}\int_0^\infty x^{s-3/2}\, {\sum_{j,k}}^{\;\prime} 
e^{-\pi(j^2+k^2)x}
\sum_{i=1}^\infty e^{-\pi i^2/Ax} \, \ud x.
\end{align*}
The first integral can be evaluated in terms of the gamma function by~\eqref{gamma2}, while the second integral
can be expressed in terms of the modified Bessel function by~\eqref{bessel1}. The result is
\begin{align*}
\pi^{-s}\Gamma(s)G(s)
&= \frac{\Gamma\left(s-\frac12\right)}{\sqrt{A}\, \pi^{s-\frac12}}
\;{\sum_{j,k}}^{\;\prime} \; \frac{1}{(j^2+k^2)^{s-\frac12}} \\
&\quad +\frac{4}{A^{\frac{s}{2}+\frac14}}\, {\sum_{j,k}}^{\;\prime} \;\; \sum_{i=1}^\infty
\left(\frac{i}{\sqrt{j^2+k^2}}\right)^{s-\frac12}
K_{s-\frac12}\left(2\pi i\sqrt{\frac{j^2+k^2}{A}}\right) \\
&= \frac{4}{\sqrt{A}} \, {\pi^{-(s-\frac12)}}\,\Gamma\left(s-\frac12\right)\, \zeta\left(s-\frac12\right)\, L_{-4}\left(s-\frac12\right) \\
&\quad + \frac{4}{A^{\frac{s}{2}+\frac14}} \sum_{N=1}^\infty \sum_{i=1}^\infty r_2(N) \left(\frac{i}{\sqrt{N}}\right)^{s-\frac12}
K_{s-\frac12}\left(2\pi i\sqrt{\frac{N}{A}}\right).
\end{align*}
On using all of the above back in~\eqref{d3} we  obtain
\begin{align}
\lefteqn{
{\sum_{i,j,k}}^{\prime} \;\;\; \frac{1}{(Ai^2+j^2+k^2)^s}} \nonumber \\
&= 2A^{-s}\zeta(2s)+4 \, \sqrt{\frac{\pi}{A}} \,\frac{\Gamma\left(s-\frac12\right)}{\Gamma(s)}\, \zeta\left(s-\frac12\right)\, L_{-4}\left(s-\frac12\right)
\nonumber \\
&\quad +   \frac{4}{A^{\frac{s}{2}+\frac14}} \, \frac{\pi^s}{\Gamma(s)}\,\sum_{N=1}^\infty \sum_{i=1}^\infty r_2(N) \left(\frac{i}{\sqrt{N}}\right)^{s-\frac12}
K_{s-\frac12}\left(2\pi i\sqrt{\frac{N}{A}}\right). \label{sumpart2}
\end{align}
The terms in~\eqref{sumpart1} involve $K_{s-1}$ Bessel functions whereas $K_{s-1/2}$ Bessel functions occur in~\eqref{sumpart2}.
That is because each application of the theta function transformation formula lowers the subscript in the resulting
Bessel function by $1/2$, due to the creation of a $x^{-1/2}$ factor in the integral. 
The theta function transformation formula is used twice (i.e., the formula is squared) in the derivation of~\eqref{sumpart1} and only once in the derivation of~\eqref{sumpart2}.

Each of~\eqref{sumpart1} and~\eqref{sumpart2} turns out to have its own advantages, as will be seen in Sections~\ref{sections1} and \ref{sections1half}.

\subsection{The alternating sum $T_2(A;s)$}
The analysis in the previous sections can be modified to handle the alternating series~\eqref{sum2} which has
the term $(-1)^{i+j+k}$ in the numerator, as follows.
Separating the terms according to whether $i=0$ or $i\neq 0$ gives
\begin{align}
T_2(A;s) &= h(s) + 2H(S) \label{T21}
\intertext{where}
h(s) &= {\sum_{j,k}}^{\prime} \;\;\; \frac{(-1)^{j+k}}{(j^2+k^2)^s} \nonumber
\intertext{and}
H(s) &= \sum_{i=1}^\infty \sum_{j=-\infty}^\infty \sum_{k=-\infty}^\infty \frac{(-1)^{i+j+k}}{(Ai^2+j^2+k^2)^s}. \nonumber
\end{align}
By a known result~\eqref{zeta2}, we have
$$
h(s) =  -4(1-2^{1-s})\zeta(s)L_{-4}(s).
$$
Next, by the integral formula for the gamma function~\eqref{gamma2} we have
\begin{align*}
\pi^{-s}\Gamma(s)H(s)
&= \int_0^\infty x^{s-1} \sum_{i=1}^\infty (-1)^ie^{-\pi A x i^2} 
\sum_{j=-\infty}^\infty \sum_{k=-\infty}^\infty (-1)^{j+k}e^{-\pi x(j^2+k^2)} \; \ud x \\
&= \int_0^\infty x^{s-1} \sum_{i=1}^\infty (-1)^ie^{-\pi A x i^2} 
\left(\sum_{j=-\infty}^\infty (-1)^je^{-\pi xj^2}\right)^2 \; \ud x.
\intertext{Now apply the modular transformation for theta functions to obtain}
\pi^{-s}\Gamma(s)H(s)
&= \int_0^\infty x^{s-1} \sum_{i=1}^\infty (-1)^ie^{-\pi A x i^2} 
\left(\frac{1}{\sqrt{x}}\sum_{j=-\infty}^\infty e^{-\pi (j+\frac12)^2/x}\right)^2 \; \ud x.
\intertext{By formula~\eqref{theta4} this can be expressed as}
\pi^{-s}\Gamma(s)H(s)&= \int_0^\infty x^{s-2} \sum_{i=1}^\infty (-1)^ie^{-\pi A x i^2} 
\sum_{N=0}^\infty r_2(4N+1) e^{-\pi (4N+1)/2x}\ \; \ud x \\
&= \sum_{i=1}^\infty\sum_{N=0}^\infty  (-1)^i r_2(4N+1) 
\int_0^\infty x^{s-2}  e^{-\pi A x i^2-\pi (4N+1)/2x}\ \; \ud x.
\intertext{The integral can be expressed in terms of Bessel functions by~\eqref{bessel1} to give}
\pi^{-s}\Gamma(s)H(s)&= 2\sum_{i=1}^\infty\sum_{N=0}^\infty  (-1)^ir_2(4N+1) \left(\frac{2N+\frac12}{Ai^2}\right)^{(s-1)/2} K_{s-1}\left(2\pi i \sqrt{A(2N+\frac12)}\right).
\end{align*}
On using all of the above back in~\eqref{T21} we obtain
\begin{align}
\lefteqn{{\sum_{i,j,k}}^{\prime} \;\;\; \frac{(-1)^{i+j+k}}{(Ai^2+j^2+k^2)^s}} \nonumber \\
&=  -4(1-2^{1-s})\zeta(s)L_{-4}(s)  \nonumber \\
&+\frac{4\pi^s}{\Gamma(s)} \, A^{(1-s)/2}\,
\sum_{i=1}^\infty\sum_{N=0}^\infty  (-1)^ir_2(4N+1) \left(\frac{2N+\frac12}{i^2}\right)^{(s-1)/2} K_{s-1}\left(2\pi i \sqrt{A(2N+\frac12)}\right).
 \label{sumpart3}
\end{align}

\subsection{A second formula for the alternating sum $T_2(A;s)$}
This time we separate the terms
according to whether $(j,k)=(0,0)$ or $(j,k) \neq (0,0)$ and write
\begin{equation}
\label{dd2}
T_2(A;s) = 2\sum_{i=1}^\infty \frac{(-1)^i}{(Ai^2)^s} + J(s)
\end{equation}
where
$$
J(s) =
{\sum_{j,k}}^{\,\prime} \; \sum_{i=-\infty}^\infty \frac{(-1)^{i+j+k}}{(Ai^2+j^2+k^2)^s}.
$$
By \eqref{zeta4} we have
$$
2\sum_{i=1}^\infty \frac{(-1)^i}{(Ai^2)^s}  = -2A^{-s}(1-2^{1-2s})\zeta(2s).
$$
It remains to analyse the sum for $J(s)$. By the integral formula for the gamma function \eqref{gamma2} we have
\begin{align*}
\pi^{-s}\Gamma(s)J(s)
&= \int_0^\infty x^{s-1} {\sum_{j,k}}^{\;\prime} (-1)^{j+j}e^{-\pi(j^2+k^2)x}
\sum_{i=-\infty}^\infty (-1)^ie^{-\pi i^2 Ax} \, \ud x.
\intertext{Now apply the modular transformation~\eqref{theta1b} to obtain}
\pi^{-s}\Gamma(s)J(s)
&= \frac{1}{\sqrt{A}}\int_0^\infty x^{s-3/2}\, {\sum_{j,k}}^{\;\prime} (-1)^{j+k}e^{-\pi(j^2+k^2)x}
\sum_{i=-\infty}^\infty e^{-\pi (i+\frac12)^2/Ax} \, \ud x.
\end{align*}
Now put $N=j^2+k^2$ and use
$$
\sum_{i=-\infty}^\infty e^{-\pi (i+\frac12)^2/Ax} = 2\sum_{i=0}^\infty e^{-\pi (i+\frac12)^2/Ax} = 2\sum_{i=1}^\infty e^{-\pi (i-\frac12)^2/Ax}
$$
to deduce 
\begin{align*}
\pi^{-s}\Gamma(s)J(s)
&= \frac{2}{\sqrt{A}} \sum_{N=1}^\infty\sum_{i=1}^\infty (-1)^Nr_2(N)  \int_0^\infty x^{s-3/2}\, e^{-\pi Nx - \pi (i-\frac12)^2/Ax} \, \ud x.
\intertext{The integral can be evaluated in terms of the modified Bessel function by~\eqref{bessel1} to give}
\pi^{-s}\Gamma(s)J(s)
&=\frac{4}{A^{\frac{s}{2}+\frac14}} \,
\sum_{N=1}^\infty \sum_{i=1}^\infty (-1)^N\,r_2(N) \left(\frac{i-\frac12}{\sqrt{N}}\right)^{s-\frac12}
K_{s-\frac12}\left(2\pi (i-\frac12)\sqrt{\frac{N}{A}}\right). 
\end{align*}
It follows that
\begin{align}
\lefteqn{{\sum_{i,j,k}}^{\prime} \;\;\; \frac{(-1)^{i+j+k}}{(Ai^2+j^2+k^2)^s}} \nonumber \\
&= -2A^{-s}(1-2^{1-{2s}})\zeta(2s) \nonumber \\
&\quad +   \frac{4}{A^{\frac{s}{2}+\frac14}} \, \frac{\pi^s}{\Gamma(s)}\,
\sum_{N=1}^\infty \sum_{i=1}^\infty (-1)^N\,r_2(N) \left(\frac{i-\frac12}{\sqrt{N}}\right)^{s-\frac12}
K_{s-\frac12}\left(2\pi (i-\frac12)\sqrt{\frac{N}{A}}\right). \label{sumpart4}
\end{align}

\subsection{Two formulas for $L(A;s)$}
On substituting the results of~\eqref{sumpart1} and \eqref{sumpart3} back into~\eqref{sum3} we obtain a formula for $L(A;s)$
in terms of~$K_{s-1}$ Bessel functions:
\begin{align}
L(A;s) 
&= 4\left(\frac{A+1}{2}\right)^s \zeta(s)L_{-4}(s) + \frac{\pi A}{s-1} \left(1+\frac{1}{A}\right)^s \zeta(2s-2) \nonumber \\
&\quad + \frac{2\pi^s\sqrt{A}}{\Gamma(s)} \left(\sqrt{A}+\frac{1}{\sqrt{A}}\right)^s \sum_{i=1}^\infty\sum_{N=1}^\infty  r_2(N) \left(\frac{N}{i^2}\right)^{(s-1)/2} K_{s-1}\left(2\pi i \sqrt{AN}\right) \nonumber \\
&\quad + \frac{2\pi^s\sqrt{A}}{\Gamma(s)} \left(\sqrt{A}+\frac{1}{\sqrt{A}}\right)^s
\sum_{i=1}^\infty\sum_{N=0}^\infty  (-1)^ir_2(4N+1)\nonumber \\
&\qquad\qquad\qquad\times  \left(\frac{2N+\frac12}{i^2}\right)^{(s-1)/2}  K_{s-1}\left(2\pi i \sqrt{A(2N+\frac12)}\right).
\label{L3formula1}
\end{align}
\noindent
On the other hand, if the results of~\eqref{sumpart2} and~\eqref{sumpart4} are used in~\eqref{sum3}, the resulting formula for $L(A;s)$ involves
$K_{s-1/2}$ Bessel functions:
\begin{align}
L(A;s)
&= 2\left(\frac{A+1}{4A}\right)^s\zeta(2s)
+2\,\sqrt{\frac{\pi}{A}}\,(A+1)^s\,\frac{\Gamma\left(s-\frac12\right)}{\Gamma(s)}\, \zeta\left(s-\frac12\right)\, L_{-4}\left(s-\frac12\right) \nonumber \\
&\quad + \frac{2}{A^{1/4}}\left(\sqrt{A}+\frac{1}{\sqrt{A}}\right)^s\, \frac{\pi^s}{\Gamma(s)}\, \,\sum_{N=1}^\infty \sum_{i=1}^\infty N^{(1-2s)/4}\,r_2(N) \nonumber \\
&\qquad \times \left\{ i^{s-\frac12}
K_{s-\frac12}\left(2\pi i\sqrt{\frac{N}{A}}\right) + (-1)^N \left(i-\frac12\right)^{s-\frac12}
K_{s-\frac12}\left(2\pi (i-\frac12)\sqrt{\frac{N}{A}}\right) \right\}.
\label{L3formula2}
\end{align}
The formulas \eqref{L3formula1} and \eqref{L3formula2} can be used as checks against each other.
Moreover, the formulas offer different information about special values of the lattice sum, as will be seen in Section~\ref{pole}.

\section{Hexagonal close packing}
Because of its importance in solid state chemistry and physics, we give a similar analysis of the lattice sum for the hexagonal close packed structure given by \cite{Kane}
$$
L_{3}^{\text{HCP}}(s) = S_1(s) + S_2(s)
$$
where
$$
S_1(s) ={\sum_{i,j,k}}^{\prime} \frac{1}{(i^2+ij+j^2+\frac83k^2)^s}
$$
and
$$
S_2(s) = \sum_{i,j,k} \frac{1}{((i+\frac13)^2+(i+\frac13)(j+\frac13)+(j+\frac13)^2
+\frac83(k+\frac12)^2)^s}.
$$
As before, the prime on the sum for $S_1(s)$ indicates that the summation is over all integers except for the term corresponding to
$i=j=k=0$ which is omitted. The sum for $S_2(s)$ is over all integer values of $i$, $j$ and $k$.
We shall analyse $S_1(s)$ and $S_2(s)$ one at a time.

\subsection{The sum $S_1(s)$}
Break the sum for $S_1(s)$ into two, according to whether $k=0$ or $k \neq 0$. This gives
\begin{equation}
S_1(s) = f(s) + 2F(s)
\end{equation}
where
$$
f(s) =   {\sum_{i,j}}^{\;\prime} \frac{1}{(i^2+ij+j^2)^s}
$$
and
$$
F(s) = \sum_{k=1}^\infty\sum_{i,j}  \frac{1}{(i^2+ij+j^2+\frac83k^2)^s}.
$$
By~\eqref{zeta7} we have
$$
f(s) = 6\zeta(s) \,L_{-3}(s)
$$
where
$$
L_{-3}(s) = \sum_{k=1}^\infty \left(\frac{\sin(2k\pi/3)}{\sin(2\pi/3)}\right) \frac{1}{k^s} = \frac{1}{1^s}-\frac{1}{2^s}+\frac{1}{4^s}-\frac{1}{5^s}+\cdots.
$$
It remains to calculate $F(s)$.
Applying the gamma function integral \eqref{gamma2} followed by the theta function transformation formula~\eqref{theta5},
we obtain
\begin{align*}
(2\pi)^{-s}\Gamma(s)F(s)
&= \int_0^\infty x^{s-1} \sum_{k=1}^\infty  e^{-16\pi k^2x/3} 
\sum_{i,j} e^{-2\pi (i^2+ij+j^2)x}
\, \ud x \\
&= \frac{1}{\sqrt{3}}\,\int_0^\infty x^{s-2} \sum_{k=1}^\infty  e^{-16\pi k^2x/3} 
\sum_{i,j} e^{-2\pi (i^2+ij+j^2)/3x}
\, \ud x.
\end{align*}
Now separate out the $i=j=0$ term and evaluate the resulting integrals. The result is
\begin{align*}
(2\pi)^{-s}\Gamma(s)F(s)
&= 
\frac{1}{\sqrt{3}}\, \sum_{k=-\infty}^\infty \int_0^\infty x^{s-2} e^{-16\pi k^2x/3} 
\, \ud x \\
&\quad +\frac{1}{\sqrt{3}}\,\sum_{k=-\infty}^\infty\;{\sum_{i,j}}^{\;\prime} 
\int_0^\infty x^{s-2}   e^{-16\pi k^2x/3-2\pi (i^2+ij+j^2)/3x} \, \ud x  \\
&=\frac{1}{\sqrt{3}} \left(\frac{3}{16\pi}\right)^{s-1} \Gamma(s-1)\zeta(2s-2)  \\
&\quad +\frac{2}{\sqrt{3}} 
\sum_{k=1}^\infty {\sum_{i,j}}^{\;\prime}  
\left(\frac{i^2+ij+j^2}{8k^2}\right)^{(s-1)/2}
K_{s-1}\left(\frac{8\pi k}{3} \sqrt{2(i^2+ij+j^2)}\right).
\end{align*}
It follows that
\begin{align}
S_1(s) &= 6\zeta(s) \,L_{-3}(s)
+\frac{4\pi}{\sqrt{3}}\left(\frac38\right)^{s-1} \left(\frac{1}{s-1}\right) \zeta(2s-2)
\nonumber \\
&\quad
+ \frac{8}{\sqrt{3}} \, \frac{\pi^s}{\Gamma(s)} 
\sum_{k=1}^\infty {\sum_{i,j}}^{\;\prime}  
\left(\frac{i^2+ij+j^2}{2k^2}\right)^{(s-1)/2}
K_{s-1}\left(\frac{8\pi k}{3} \sqrt{2(i^2+ij+j^2)}\right) \nonumber \\
&= 6\zeta(s) \,L_{-3}(s)
+\frac{4\pi}{\sqrt{3}}\left(\frac38\right)^{s-1} \left(\frac{1}{s-1}\right) \zeta(2s-2)
\nonumber \\
&\quad
+ \frac{8}{\sqrt{3}} \, \frac{\pi^s}{\Gamma(s)} 
\sum_{k=1}^\infty \sum_{N=1}^\infty u_2(N)
\left(\frac{N}{2k^2}\right)^{(s-1)/2}
K_{s-1}\left(\frac{8\pi k}{3} \sqrt{2N}\right) \label{S11}
\end{align}
where $u_2(N)$ is the number of representations of $N$ by the form $i^2+ij+j^2$.

\subsection{A second formula for the sum $S_1(s)$}
A different formula for $S_1(s)$ can be obtained by separating the terms in the series according to whether $i=j=0$ or $i$ and $j$ are
not both zero. This gives
$$
S_1(s) = 2 \left(\frac38\right)^s \sum_{k=1}^\infty \frac{1}{k^{2s}} + G(s)
$$
where
$$
G(s) = {\sum_{i,j}}^{\;\prime} \sum_{k=-\infty}^\infty \frac{1}{(i^2+ij+j^2+\frac83k^2)^s}.
$$
Applying the gamma function integral \eqref{gamma2} followed by the theta function transformation formula~\eqref{theta1},
we obtain
\begin{align*}
\pi^{-s}\Gamma(s)G(s)
&= \int_0^\infty x^{s-1} {\sum_{i,j}}^{\;\prime} e^{-\pi (i^2+ij+j^2)x}
\sum_{k=-\infty}^\infty e^{-8\pi k^2x/3} \, \ud x \\
&= \sqrt{\frac38} \int_0^\infty x^{s-\frac32}
 {\sum_{i,j}}^{\;\prime} e^{-\pi (i^2+ij+j^2)x}
\sum_{k=-\infty}^\infty e^{-3\pi k^2/8x} \, \ud x.
\end{align*}
Now separate out the $k=0$ term and evaluate the resulting integrals. The result is
\begin{align*}
\lefteqn{\pi^{-s}\Gamma(s)G(s)} \\
&= 
\sqrt{\frac38} \int_0^\infty x^{s-\frac32}
 {\sum_{i,j}}^{\;\prime} e^{-\pi (i^2+ij+j^2)x} \, \ud x \\
&\quad + 2\sqrt{\frac38} \int_0^\infty x^{s-\frac32}
 {\sum_{i,j}}^{\;\prime} e^{-\pi (i^2+ij+j^2)x}
\sum_{k=1}^\infty e^{-3\pi k^2/8x} \, \ud x \\
&= \sqrt{\frac38}\,\pi^{-(s-\frac12)}\,\Gamma(s-\frac12)
{\sum_{i,j}}^{\;\prime} \frac{1}{(i^2+ij+j^2)^{s-1/2}} \\
&\quad +4\left(\frac38\right)^{(2s+1)/4}
{\sum_{i,j}}^{\;\prime} \sum_{k=1}^\infty \left(\frac{k^2}{i^2+ij+j^2}\right)^{(2s-1)/4}
K_{s-\frac12}\left(\sqrt\frac32\,\pi k \sqrt{i^2+ij+j^2} \right).
\end{align*}
The first sum can be evaluated in terms of the Riemann zeta function and the $L_{-3}$ function by~\eqref{zeta7}.
In the second sum, we again use the notation $u_2(N)$ for the number of representations of $N$ by the form $i^2+ij+j^2$. The result is
\begin{align*}
\pi^{-s}\Gamma(s)G(s)
&=  \sqrt{\frac{27}2}\,\pi^{-(s-\frac12)}\,\Gamma\left(s-\frac12\right) \zeta\left(s-\frac12\right)L_{-3}\left(s-\frac12\right) \\
&\quad +4\left(\frac38\right)^{(2s+1)/4}
\sum_{N=1}^\infty \sum_{k=1}^\infty u_2(N)\, \left(\frac{k^2}{N}\right)^{(2s-1)/4}
K_{s-\frac12}\left(\pi k \sqrt{\frac{3N}{2}} \right).
\end{align*}
It follows that
\begin{align}
S_1(s) &= 2\left(\frac38\right)^s \zeta(2s)
+\sqrt\frac{27\pi}{2} \;\frac{\Gamma(s-\frac12)}{\Gamma(s)}\;
\zeta\left(s-\frac12\right)L_{-3}\left(s-\frac12\right)
\nonumber \\
&\quad
+ \frac{4\pi^s}{\Gamma(s)}\,\left(\frac38\right)^{(2s+1)/4} 
\sum_{N=1}^\infty \sum_{k=1}^\infty u_2(N)\,  \left(\frac{k^2}{N}\right)^{(2s-1)/4}
K_{s-\frac12}\left(\pi k \sqrt{\frac{3N}{2}} \right). \label{S12}
\end{align}

\subsection{The sum $S_2(s)$} The analysis in this case is a little simpler because
the summation is over all integers $i$, $j$ and $k$.
We apply the gamma function integral~\eqref{gamma2} 
to write
\begin{align}
S_2(s) &= \sum_{i,j,k} \frac{1}{((i+\frac13)^2+(i+\frac13)(j+\frac13)+(j+\frac13)^2
+\frac83(k+\frac12)^2)^s} \nonumber \\
&= \frac{(2\pi)^s}{\Gamma(s)} 
\int_0^\infty x^{s-1}
  \sum_{k=-\infty}^\infty e^{-16\pi(k+\frac12)^2x/3}
 \sum_{i,j=-\infty}^\infty e^{-2\pi ((i+\frac13)^2+(i+\frac13)(j+\frac13)+(j+\frac13)^2)x}
 \, \ud x. \label{S2a}
 \end{align}
Now make use of the transformation formula ~\eqref{theta5a}
to deduce
\begin{align*}
S_2(s)
&= \frac{(2\pi)^s}{\Gamma(s)} 
\int_0^\infty x^{s-1}
 \left(2 \sum_{k=0}^\infty e^{-16\pi(k+\frac12)^2x/3}\right)
\left(\frac{1}{x\sqrt{3}} \sum_{i,j=-\infty}^\infty \omega^{i-j}e^{-2\pi (i^2+ij+j^2)/3x}\right)
 \, \ud x \\
\intertext{where $\omega=e^{2\pi i/3}$ is a primitive cube root of 1. Now separate the term $i=j=0$ to deduce}
S_2(s) 
&= \frac{(2\pi)^s}{\Gamma(s)} \frac{2}{\sqrt{3}}
\int_0^\infty x^{s-2}
  \sum_{k=0}^\infty e^{-16\pi(k+\frac12)^2x/3}
 \, \ud x \\
&\quad +\frac{(2\pi)^s}{\Gamma(s)} \frac{2}{\sqrt{3}}
\int_0^\infty x^{s-2}
  \sum_{k=0}^\infty e^{-16\pi(k+\frac12)^2x/3}\;\;
\sum_{N=1}^\infty  \omega^N\, u_2(N)\,  e^{-2\pi N/3x}
 \, \ud x
 \end{align*}
 where $u_2(N)$ is the number of representations of $N$ by the form $i^2+ij+j^2$, as before.
 On evaluating the integrals using~\eqref{gamma2} and~\eqref{bessel1} we obtain
\begin{align*}
S_2(s)
&= \frac{4\pi}{\sqrt{3}}\left(\frac38\right)^{s-1} \left(\frac{1}{s-1}\right)\,
\sum_{k=0}^\infty \frac{1}{(k+\frac12)^{2s-2}}\\
&\quad + \frac{8}{\sqrt{3}} \, \frac{\pi^s}{\Gamma(s)} 
\sum_{k=0}^\infty \sum_{N=1}^\infty  \omega^N \,u_2(N)\,
\left(\frac{N}{2(k+\frac12)^2}\right)^{(s-1)/2}
K_{s-1}\left(\frac{8\pi \left(k+\frac12\right) }{3} \sqrt{2N}\right).
\end{align*}
The first sum can be evaluated in terms of the Riemann zeta function by using~\eqref{zeta4a}.
The $\omega^N$ term can be handled by using
$$
\omega^N = \cos\frac{2\pi N}{3} + i \sin\frac{2\pi N}{3}
$$
along with the fact that $S_2(s)$ is real valued when $s$ is real. It follows that
\begin{align}
S_2(s) &=\frac{4\pi}{\sqrt{3}}\left(\frac38\right)^{s-1}(2^{2s-2}-1) \left(\frac{1}{s-1}\right) \zeta(2s-2)\nonumber \\
&\quad + \frac{8}{\sqrt{3}} \, \frac{\pi^s}{\Gamma(s)} 
\sum_{k=0}^\infty \sum_{N=1}^\infty  \cos\frac{2\pi N}{3} \,u_2(N)\,
\left(\frac{N}{2(k+\frac12)^2}\right)^{(s-1)/2}
K_{s-1}\left(\frac{8\pi \left(k+\frac12\right) }{3} \sqrt{2N}\right).
\label{S21}
\end{align}

\subsection{A second formula for the sum $S_2(s)$} 
We introduce the abbreviation
$$
Y_{i,j}= \left(i+\frac13\right)^2+\left(i+\frac13\right)\left(j+\frac13\right)+\left(j+\frac13\right)^2
$$
to write~\eqref{S2a} in the form
\begin{align*}
S_2(s) 
&= \frac{(2\pi)^s}{\Gamma(s)} 
\int_0^\infty x^{s-1}
 \sum_{i,j=-\infty}^\infty e^{-2\pi Y_{i,j}x}
  \sum_{k=-\infty}^\infty e^{-16\pi(k+\frac12)^2x/3}
 \, \ud x.
 \intertext{This time we apply the transformation formula~\eqref{theta1b} to the sum over~$k$ to obtain}
S_2(s) &= \frac{\sqrt{3}}{4}\,\frac{(2\pi)^s}{\Gamma(s)} 
\int_0^\infty x^{s-3/2}
 \sum_{i,j=-\infty}^\infty e^{-2\pi Y_{i,j}x}
  \sum_{k=-\infty}^\infty (-1)^k\,e^{-3\pi k^2/16x}
 \, \ud x.
 \end{align*}
Now separate the terms according to whether $k=0$ or $k \neq 0$ and evaluate the resulting integrals by~\eqref{gamma2} and~\eqref{bessel1}.
The result is
\begin{align*}
S_2(s)
&= \sqrt{\frac{3\pi}{8}} \, \frac{\Gamma(s-\frac12)}{\Gamma(s)} \,
\sum_{i,j=-\infty}^\infty \frac{1}{Y_{ij}^{s-1/2}}\\
& \quad + \frac{4\pi^s}{\Gamma(s)} \left(\frac38\right)^{(2s+1)/4}
\sum_{k=1}^\infty (-1)^k \sum_{i,j=-\infty}^\infty \left(\frac{k}{\sqrt{Y_{ij}}}\right)^{s-\frac12}
K_{s-\frac12}\left(\pi k \sqrt{3Y_{i,j}/2}\right).
\end{align*}
The first sum can be handled by~\eqref{zeta8} to give
$$\sum_{i,j=-\infty}^\infty \frac{1}{Y_{ij}^{s-1/2}}
= 3(3^{s-1/2}-1)\zeta\left(s-\frac12\right) L_{-3}\left(s-\frac12\right).
$$
For the other sum, observe that
$$
3Y_{i,j} = 3i^2+3ij+3j^2+3i+3j+1,
$$
that is to say $3Y_{i,j}$ is a positive integer and $3Y_{i,j} \equiv 1 \pmod{3}$. Therefore we set $3Y_{i,j}=3N+1$ and use~\eqref{theta2b}
to deduce that
the number of solutions of 
$$
3i^2+3ij+3j^2+3i+3j+1=3N+1
$$
is equal to $\frac12u_2(3N+1)$.

Taking all of the above into account, we finally obtain
\begin{align}
S_2(s) &= \sqrt\frac{27\pi}{8} \;\frac{\Gamma(s-\frac12)}{\Gamma(s)}\;(3^{s-1/2}-1)\;
\zeta\left(s-\frac12\right)L_{-3}\left(s-\frac12\right) \nonumber \\
& \quad + \frac{2\pi^s}{\Gamma(s)} \left(\frac38\right)^{(2s+1)/4}
\sum_{k=1}^\infty  \sum_{N=0}^\infty(-1)^k\, u_2(3N+1) \nonumber \\
&\qquad\qquad\qquad \times \left(\frac{k}{\sqrt{N+\frac13}}\right)^{s-\frac12}
K_{s-\frac12}\left(\pi k \sqrt{\frac{3N+1}{2}}\right).
\label{S22}
\end{align}

\subsection{The lattice sum for hexagonal close packing}
On adding the results for $S_1(s)$ and $S_2(s)$ in~\eqref{S11} and~\eqref{S21} we obtain
\begin{align}
L_{3}^{\text{HCP}}(s)  
&= 6\zeta(s) \,L_{-3}(s)
+\frac{4\pi}{\sqrt{3}}\left(\frac32\right)^{s-1} \left(\frac{1}{s-1}\right) \zeta(2s-2)
\nonumber \\
&\quad
+ \frac{8}{\sqrt{3}} \, \frac{\pi^s}{\Gamma(s)} 
\sum_{k=1}^\infty \sum_{N=1}^\infty u_2(N)
\left(\frac{N}{2k^2}\right)^{(s-1)/2}
K_{s-1}\left(\frac{8\pi k}{3} \sqrt{2N}\right) \nonumber \\
&\quad + \frac{8}{\sqrt{3}} \, \frac{\pi^s}{\Gamma(s)} 
\sum_{k=0}^\infty \sum_{N=1}^\infty  \cos\frac{2\pi N}{3} \,u_2(N)\,\nonumber \\
&\qquad\qquad\qquad \times\left(\frac{N}{2(k+\frac12)^2}\right)^{(s-1)/2}
K_{s-1}\left(\frac{8\pi \left(k+\frac12\right) }{3} \sqrt{2N}\right).
\label{hcp1}
\end{align}

\noindent
On the other hand, if we add the results of~\eqref{S12} and~\eqref{S22} we obtain
\begin{align}
L_{3}^{\text{HCP}}(s)  
&=2\left(\frac38\right)^s \zeta(2s)
+ \sqrt\frac{27\pi}{8} \;\frac{\Gamma(s-\frac12)}{\Gamma(s)}\;(3^{s-1/2}+1)\;
\zeta\left(s-\frac12\right)L_{-3}\left(s-\frac12\right)
\nonumber \\
&\quad
+ \frac{4\pi^s}{\Gamma(s)}\,\left(\frac38\right)^{(2s+1)/4} 
\sum_{N=1}^\infty \sum_{k=1}^\infty u_2(N)\,  \left(\frac{k}{\sqrt{N}}\right)^{s-1/2}
K_{s-\frac12}\left(\pi k \sqrt{\frac{3N}{2}} \right)
\nonumber \\
& \quad + \frac{2\pi^s}{\Gamma(s)} \left(\frac38\right)^{(2s+1)/4}
\sum_{k=1}^\infty  \sum_{N=0}^\infty(-1)^k\, u_2(3N+1) \nonumber \\
&\qquad\qquad \qquad \times \left(\frac{k}{\sqrt{N+\frac13}}\right)^{s-\frac12}
K_{s-\frac12}\left(\pi k \sqrt{\frac{3N+1}{2}}\right)
\label{hcp2}
\end{align}

\section{Analytic continuations of the lattice sums $L(A;s)$ and $L_{3}^{\text{HCP}}(s)$}
\label{pole}
We will now show that the lattice sums $L(A;s)$ and $L_{3}^{\text{HCP}}(s)$
can be continued analytically to the whole $s$-plane, and that the resulting functions have a single simple pole
at $s=3/2$ and no other singularities. We do this in steps. First, in Section~\ref{6point1} we show that
the lattice sums each have a simple pole at $s=3/2$ and determine the residue. Then, in Section~\ref{6point2}
we show that the analytic continuations obtained are valid for the whole $s$-plane and there are no other singularities.
Finally, in Sections~\ref{sections1half}--\ref{sequals0}, values of the analytic continuations at the points $s=1/2$ and $s=1,\,0,\,-1,\,-2,\ldots$
are computed. In particular, the evaluation of $T_2(A;s)$ at $s=1/2$ in the case $A=1$ gives the Madelung constant,
e.g., see~\cite{borwein1998convergence},
~\cite[pp. xiii, 39--51]{BorweinEtAl},~\cite{madelung1918}.
\subsection{Behaviour of the lattice sums at $s=3/2$}
\label{6point1}
We start by showing that $L(A;s)$ has a simple pole at $s=3/2$ and determine the residue.
In the formula~\eqref{L3formula1}, all of the terms are analytic at $s=3/2$ except for the term involving $\zeta(2s-2)$.
It follows that
\begin{align*}
\lim_{s\rightarrow 3/2} (s-3/2)L(A;s) 
&= \lim_{s\rightarrow 3/2} (s-3/2)\frac{\pi A}{s-1} \left(1+\frac{1}{A}\right)^s \zeta(2s-2)  \\
&=2\pi A \left(1+\frac{1}{A}\right)^{3/2} \lim_{s\rightarrow 3/2} (s-3/2) \zeta(2s-2) \\
&=  \frac{2\pi}{\sqrt{A}} \left(A+1\right)^{3/2} \times \frac12 \; \lim_{u\rightarrow 1} (u-1) \zeta(u) \\
&= \frac{\pi}{\sqrt{A}}  \left(A+1\right)^{3/2}
\end{align*}
where~\eqref{zetapole} was used in the last step of the calculation.
It follows further that $L(A;s)$ has a simple pole at $s=3/2$ and the residue is given by
$$
\text{Res}(L(A;s),3/2) = \frac{\pi}{\sqrt{A}}  \left(A+1\right)^{3/2}.
$$
By~\eqref{pdA} this is just 12 times the packing density, i.e.,
$$
\text{Res}(L(A;s),3/2) = 12\Delta_\mathcal{L}.
$$
For example, taking $A=1$ gives
\begin{equation}
\label{residue1}
\text{Res}(L_3^{\text{FCC}}(s),3/2)  = 2\sqrt{2}\,\pi
\end{equation}
while taking $A=1/2$ gives
 $$
\text{Res}(L_3^{\text{BCC}}(s),3/2)  = 3\sqrt{3}\,\pi/2.
$$
Laurent's theorem implies there is an expansion of the form
\begin{equation}
\label{laurent1}
L(A;s) = \frac{c_{-1}}{s-3/2}+c_0+\sum_{n=1}^\infty c_n (s-3/2)^n
\end{equation}
where
$$
c_{-1} = \text{Res}(L(A;s),3/2) = \frac{\pi}{\sqrt{A}}  \left(A+1\right)^{3/2}
$$
and the coefficients $c_0$, $c_1$, $c_2,\ldots$ depend on $A$ but not on $s$.
To calculate~$c_0$, start with the fact that
$$
\lim_{s\rightarrow 3/2} \left(\frac{\pi A}{s-1} \left(1+\frac{1}{A}\right)^s \zeta(2s-2) - \frac{c_{-1}}{s-3/2}\right) 
$$
$$
= \frac{\pi}{\sqrt{A}} \left(A+1\right)^{3/2} \left(2\gamma - 2 + \log\left(1+\frac{1}{A}\right)\right)
$$
where $\gamma = 0.57721\,56649\,01532\,86060\,\cdots$ is Euler's constant. Then use~\eqref{L3formula1} and~\eqref{bessel2} to deduce
\begin{align*}
c_0 &= \lim_{s\rightarrow 3/2} \left( L(A;s) - \frac{c_{-1}}{s-3/2}\right) \\
&=\sqrt{2}\left(A+1\right)^{3/2} \zeta\left(\frac32\right)L_{-4}\left(\frac32\right) \\
&\quad + \frac{\pi}{\sqrt{A}} \left(A+1\right)^{3/2} \left(2\gamma - 2 + \log\left(1+\frac{1}{A}\right) \right) \\
&\quad + \frac{2\pi}{\sqrt{A}} \left(A+1\right)^{3/2} \sum_{k=1}^\infty\sum_{N=1}^\infty  \frac{1}{k} \, r_2(N) \exp\left(-2\pi k \sqrt{AN}\right) \\
&\quad + \frac{2\pi}{\sqrt{A}} \left(A+1\right)^{3/2} \sum_{k=1}^\infty\sum_{N=0}^\infty  \frac{(-1)^k}{k}\,r_2(4N+1)\, \exp\left(-2\pi k \sqrt{A\left(2N+\frac12\right)}\right).
\intertext{Interchanging the order of summation and evaluating the sum over~$k$ gives}
c_0 &=\sqrt{2}\left(A+1\right)^{3/2} \zeta\left(\frac32\right)L_{-4}\left(\frac32\right) \\
&\quad + \frac{\pi}{\sqrt{A}} \left(A+1\right)^{3/2} \left(2\gamma - 2 + \log\left(1+\frac{1}{A}\right) \right) \\
&\quad - \frac{2\pi}{\sqrt{A}} \left(A+1\right)^{3/2} \sum_{N=1}^\infty   r_2(N) \log\left(1- e^{-2\pi \sqrt{AN}}\right) \\
&\quad - \frac{2\pi}{\sqrt{A}} \left(A+1\right)^{3/2} \sum_{N=0}^\infty  r_2(4N+1)\, \log\left(1+ e^{-\pi  \sqrt{2A\left(4N+1\right)}}\right).
\end{align*}
Numerical evaluation in the case $A=1$ gives
\begin{equation}
\label{c0}
\left. c_0\right|_{A=1} = 6.98405\,25503\,22247\,93406\cdots.
\end{equation}
A similar analysis can be given for $L_3^{\text{HCP}}(s)$ using~\eqref{hcp1}. We omit the details of the calculations as they are similar to the above.
The end result is a Laurent expansion of the form
\begin{equation}
\label{laurent2}
L_3^{\text{HCP}}(s) = \frac{d_{-1}}{s-3/2}+d_0+\sum_{n=1}^\infty d_n (s-3/2)^n
\end{equation}
where
\begin{equation}
\label{residue2}
d_{-1} = \text{Res}(L_3^{\text{HCP}}(s),3/2) = 2\sqrt{2}\pi
\end{equation}
and
\begin{align}
d_0
&= 6\zeta\left(\frac32\right) L_{-3}\left(\frac32\right) \;+ \; 2\,\sqrt{2}\,\pi \left(2\gamma-2+\log\frac32\right) \nonumber \\
&\quad
+ 2\sqrt{2}\pi \sum_{k=1}^\infty \sum_{N=1}^\infty \frac{1}{k}\,u_2(N)
\exp\left(-\frac83\pi k \sqrt{2N}\right) \nonumber \\
&\quad + 2\sqrt{2}\pi 
\sum_{k=0}^\infty \sum_{N=1}^\infty  \frac{\cos\left(\frac{2\pi N}{3}\right)}{(k+\frac12)} \,u_2(N)\,
\exp\left(-\frac83\pi \left(k+\frac12\right)\, \sqrt{2N}\right) \nonumber \\
&= 6\zeta\left(\frac32\right) L_{-3}\left(\frac32\right) \; + \; 2\,\sqrt{2}\,\pi \left(2\gamma-2+\log\frac32\right) \nonumber \\
&\quad
- 2\sqrt{2}\pi  \sum_{N=1}^\infty u_2(N) \log\left(1-e^{-\frac83\pi  \sqrt{2N}}\right) \nonumber \\
&\quad + 2\sqrt{2}\pi  \sum_{N=1}^\infty  \cos\left(\frac{2\pi N}{3}\right) \,u_2(N)\,
\log\left(\frac{1+e^{-\frac43\pi \sqrt{2N}}}{1-e^{-\frac43\pi \sqrt{2N}}}\right) \nonumber \\
&=  6.98462\;37414\;38416\;61307\;\cdots.
\label{d0}
\end{align}
In particular, the pole of $L_3^{\text{HCP}}(s)$ at $s=3/2$ is simple.
By~\eqref{residue1} and~\eqref{residue2} we have
$$
\text{Res}(L_3^{\text{HCP}}(s),3/2)= \text{Res}(L_3^{\text{FCC}}(s),3/2).
$$
It follows that the difference
\mbox{$L_3^{\text{FCC}}(s) - L_{3}^{\text{HCP}}(s)$}
has a removable singularity at $s=3/2$
and from the Laurent expansions we deduce that
$$
\lim_{s\rightarrow \frac32} \left(L_3^{\text{FCC}}(s) - L_{3}^{\text{HCP}}(s) \right) = \left.c_0\right|_{A=1} - d_0.
$$
Using the numerical values from~\eqref{c0} and~\eqref{d0} we obtain
$$
\lim_{s\rightarrow \frac32} \left(L_3^{\text{FCC}}(s) - L_{3}^{\text{HCP}}(s) \right) = -0.00057\;11911\;16168\;67901\cdots.
$$
This gives the value at the left hand end of the graph in \cite[Fig. 3]{paper1}. The value $s=3/2$ used here corresponds to taking $s=3$ in~\cite{paper1}
because of the different way the exponents are used in the definitions.

\subsection{Analyticity of the lattice sums at other values of $s$}
\label{6point2}
By~\eqref{bessel3}, the double series of Bessel functions in~\eqref{L3formula1} converges absolutely and uniformly on compact subsets of the $s$-plane
and therefore represents an entire function of~$s$. It follows that $L(A;s)$ has an analytic continuation to a meromorphic function
which is analytic except possibly at the singularities of the terms
\begin{equation}
\label{sing1}
4\left(\frac{A+1}{2}\right)^s \zeta(s)L_{-4}(s)
\end{equation}
and
\begin{equation}
\label{sing2}
\frac{\pi A}{s-1} \left(1+\frac{1}{A}\right)^s \zeta(2s-2).
\end{equation}
The function in \eqref{sing1} is analytic except at $s=1$ due to the pole of~$\zeta(s)$, as the function $L_{-4}(s)$ and the exponential function are both entire.
The function in~\eqref{sing2} is
analytic except at $s=1$ and $s=3/2$. We studied the singularity at $s=3/2$ in the previous section, so this
leaves only the point $s=1$. Using~\eqref{zetapole} and the values of $\zeta(0)$ and $L_{-4}(1)$ in~\eqref{zetavalues} and~\eqref{L4values} we find that
$$
4\left(\frac{A+1}{2}\right)^s \zeta(s)L_{-4}(s)  = \frac{(A+1)\pi}{2(s-1)}+O(1) \quad\text{as $s\rightarrow 1$}
$$
and 
$$
\frac{\pi A}{s-1} \left(1+\frac{1}{A}\right)^s \zeta(2s-2)= -\frac{(A+1)\pi}{2(s-1)}+O(1) \quad\text{as $s\rightarrow 1$}.
$$
It follows that the sum of the functions in~\eqref{sing1} and~\eqref{sing2} has a removable singularity at~$s=1$ and
so $L(A;s)$ is also analytic at~$s=1$. The analyticity at $s=1$ can also be seen directly from the alternative formula for~$L(A;s)$ in~\eqref{L3formula2}.

In conclusion, it has been shown that $L(A;s)$ has an analytic continuation to a meromorphic function of~$s$ which has a simple pole at~$s=3/2$
and no other singularities. Because $L(A;s)$ has only one singularity, namely $s=3/2$, the Laurent expansion~\eqref{laurent1} is valid
in the annulus $0<|s-3/2|<\infty$, i.e., for all $s \neq 3/2$. 

In a similar way,~\eqref{hcp1} and~\eqref{hcp2} can be used to show that~$L_{3}^{\text{HCP}}(s)$ also
has an analytic continuation to a meromorphic function of~$s$ which has a simple pole at~$s=3/2$
and no other singularities. The Laurent expansion~\eqref{laurent2} converges for all $s\neq 3/2$.

By the theory of complex variables, the analytic continuation, if one exists, is unique, e.g., see~\cite[p. 147, Th. 1]{levinson}.
Therefore analytic continuation formulas can be used to assign values to divergent series.
For example, the Madelung constant is defined by
\begin{equation}
\label{madelungdefinition}
M={\sum_{i,j,k}}^{\prime} \;\;\;\left. \frac{(-1)^{i+j+k}}{(i^2+j^2+k^2)^{s}}\right|_{s=1/2}. 
\end{equation}
This is interpreted as being the value of the analytic continuation of the series at $s=1/2$, because 
$$
{\sum_{i,j,k}}^{\prime} \;\;\; \frac{(-1)^{i+j+k}}{(i^2+j^2+k^2)^{s}}
$$
obviously diverges if $s=1/2$. From now on, we shall use the expression ``the value of a series at a point~$s$'' to mean ``the value of the
analytic continuation of the series at the point~$s$''.

\subsection{Values at $s=1/2$ and the Madelung constant}
\label{sections1half}
On putting $s=1/2$ in~\eqref{sumpart3} we obtain an analytic expression for the value of 
\begin{align*}
\lefteqn{M(A)={\sum_{i,j,k}}^{\prime} \;\;\;\left. \frac{(-1)^{i+j+k}}{(Ai^2+j^2+k^2)^{s}}\right|_{s=1/2}}
\intertext{which specialises to the Madelung constant in the case $A=1$. We have}
M(A)&=  -4(1-2^{1-s})\zeta(s)L_{-4}(s) \Bigg|_{s=1/2} \nonumber \\
&+\frac{4\pi^s}{\Gamma(s)} \,A^{(1-s)/2}
\sum_{i=1}^\infty\sum_{N=0}^\infty  
(-1)^ir_2(4N+1) \left(\frac{2N+\frac12}{i^2}\right)^{(s-1)/2} \\
&\qquad\qquad \qquad\qquad\qquad \times K_{s-1}\left(2\pi i \sqrt{A(2N+\frac12)}\,\right) \Bigg|_{s=1/2}.
\end{align*}
Now use~\eqref{bessel1.5} and \eqref{bessel2} to express the Bessel functions in terms of exponential functions. The result simplifies to 
\begin{align*}
M(A)&=4(\sqrt2-1)\zeta\left(\frac12\right)L_{-4}\left(\frac12\right)
+2\sum_{i=1}^\infty\sum_{N=0}^\infty  (-1)^i\, \frac{r_2(4N+1)}{\sqrt{2N+\frac12}}\,e^{-2\pi i \sqrt{A(2N+1/2)}}.
\intertext{On interchanging the order of summation and summing the geometric series, we obtain}
M(A)&=4(\sqrt2-1)\zeta\left(\frac12\right)L_{-4}\left(\frac12\right)
-2\sqrt{2}\sum_{N=0}^\infty  \frac{r_2(4N+1)}{\sqrt{4N+1}}\left(\frac{1}{e^{\pi  \sqrt{2A(4N+1)}}+1}\right).
\end{align*}
When $A=1$ this gives the Madelung constant defined by~\eqref{madelungdefinition}.
Numerical evaluation gives
\begin{equation}
\label{Mminus}
M=M(1)=-1.74756\;45946\;33182\;19063\cdots
\end{equation}
which is in agreement with~\cite[p. xiii]{BorweinEtAl} (apart from the minus sign which we have corrected here) and matches the value of $d(1)$
in~\cite[pp 39--51]{BorweinEtAl}.

\bigskip
In a similar way, starting from~\eqref{sumpart1} and using~\eqref{bessel2} and~\eqref{zetavalues} we obtain
\begin{align}
\lefteqn{{\sum_{i,j,k}}^{\prime} \;\;\;\left. \frac{1}{(Ai^2+j^2+k^2)^{s}}\right|_{s=1/2} } \nonumber \\
&=4 \zeta\left(\frac12\right)L_{-4}\left(\frac12\right)+\frac{\pi\sqrt{A}}{3} +2\sum_{i=1}^\infty\sum_{N=1}^\infty  \frac{r_2(N)}{\sqrt{N}}\,e^{-2\pi i \sqrt{AN}} \nonumber \\
&=4 \zeta\left(\frac12\right)L_{-4}\left(\frac12\right)+\frac{\pi\sqrt{A}}{3} +2\sum_{N=1}^\infty  \frac{r_2(N)}{\sqrt{N}}\left(\frac{1}{e^{2\pi \sqrt{AN}}-1}\right).
\label{remarkable1}
\end{align}
Numerical evaluation in the case $A=1$ gives
\begin{equation}
{\sum_{i,j,k}}^{\prime} \;\;\;\left. \frac{1}{(i^2+j^2+k^2)^{s}}\right|_{s=1/2} = -2.83729\;74794\;80619\;47666\cdots. \label{remarkable3}
\end{equation}
Now, from~\eqref{sum3} we have
\begin{align*}
L_{3}^{\text{FCC}}(1/2)&= \frac{1}{\sqrt{2}}{\sum_{i,j,k}}^{\prime} \left. \frac{1}{(i^2+j^2+k^2)^{s}}\right|_{s=1/2} 
+  \frac{1}{\sqrt{2}}{\sum_{i,j,k}}^{\prime} \left. \frac{(-1)^{i+j+k}}{(i^2+j^2+k^2)^{s}}\right|_{s=1/2}.
\intertext{Hence, using the values from~\eqref{Mminus} and~\eqref{remarkable3} we obtain}
L_{3}^{\text{FCC}}(1/2)&= -3.24198\;70634\;10888\;39428\cdots.
\end{align*}
We also record the result
\begin{align*}
L_{3}^{\text{HCP}}(1/2)
&= 6\zeta\left(\frac12\right)L_{-3}\left(\frac12\right) + \frac{2\sqrt{2}\pi}{9}
+2\sum_{N=1}^\infty \frac{u_2(N)}{\sqrt{N}}\left(\frac{1}{e^{8\pi\sqrt{2N}/3}-1}\right) \\
&\qquad +2\sum_{N=1}^\infty  \cos \left(\frac{2\pi N}{3}\right) \frac{u_2(N)}{\sqrt{N}}\left(\frac{1}{e^{4\pi\sqrt{2N}/3}-e^{-4\pi\sqrt{2N}/3}}\right) \\
&= -3.24185\;86150\;75732\;86473\cdots
\end{align*}
which is obtained in the same way, starting from~\eqref{hcp1}.

\subsection{The value at $s=1$}
\label{sections1}
It was noted above that~\eqref{L3formula1}, which involves $K_{s-1}$ Bessel functions, contains terms with
singularities at $s=1$ and therefore is not suitable for calculations at that value of~$s$.
Instead we can use~\eqref{L3formula2}, which involves $K_{s-1/2}$ Bessel functions.
As in the previous section, two steps are involved. First, the the $K_{1/2}$ Bessel functions can be expressed in terms of the exponential function by~\eqref{bessel2}. Then,
the double sum can be reduced to a single sum by geometric series. We omit the details and just record the final results and corresponding numerical
values.

\noindent
From~\eqref{sumpart2} we have
\begin{align}
\lefteqn{{\sum_{i,j,k}}^{\prime} \;\;\;\left. \frac{1}{(Ai^2+j^2+k^2)^{s}}\right|_{s=1}} \nonumber \\
&=\frac{\pi^2}{3A} + \frac{4\pi}{\sqrt{A}} \zeta\left(\frac12\right) L_{-4}\left(\frac12\right)
+\frac{2\pi}{\sqrt{A}} \sum_{N=1}^\infty \frac{r_2(N)}{\sqrt{N}} \left(\frac{1}{e^{2\pi \sqrt{N/A}}-1}\right) \label{remarkable2}
\intertext{while~\eqref{sumpart4}  gives}
\lefteqn{{\sum_{i,j,k}}^{\prime} \;\;\;\left. \frac{(-1)^{i+j+k}}{(Ai^2+j^2+k^2)^{s}}\right|_{s=1}} \nonumber\\
&=\frac{-\pi^2}{6A}  +\frac{2\pi}{\sqrt{A}} \sum_{N=1}^\infty (-1)^N\, \frac{r_2(N)}{\sqrt{N}} \left(\frac{1}{e^{\pi \sqrt{N/A}}-e^{-\pi \sqrt{N/A}}}\right).
\nonumber
\end{align}
Then~\eqref{sum3} can be used to write down the value of $L(A;s)$.

For example, when $A=1$ the above formulas give
\begin{align}
{\sum_{i,j,k}}^{\prime} \;\;\;\left. \frac{1}{(i^2+j^2+k^2)^{s}}\right|_{s=1}
&=-8.91363\;29175\;85151\;27268\cdots \label{remarkable4}
\intertext{and} \nonumber
{\sum_{i,j,k}}^{\prime} \;\;\;\left. \frac{(-1)^{i+j+k}}{(i^2+j^2+k^2)^{s}}\right|_{s=1}
&=-2.51935\;61520\;89445\;31334\cdots. \label{unremarkable2}
\end{align}
Then taking $A=1$ and $s=1$ in~\eqref{sum3} gives
\begin{align*}
L_3^{\text{FCC}}(1) &={\sum_{i,j,k}}^{\prime} \;\;\;\left. \frac{1}{(i^2+j^2+k^2)^{s}}\right|_{s=1} + {\sum_{i,j,k}}^{\prime} \;\;\;\left. \frac{(-1)^{i+j+k}}{(i^2+j^2+k^2)^{s}}\right|_{s=1} \\
&=  -11.43298\;90696\;74596\;58602\cdots.
\end{align*}
For HCP, the formula~\eqref{hcp1} cannot be used to evaluate~$L_3^{\text{HCP}}(1)$ because two of the terms have cancelling singularities at~$s=1$.
Therefore we take~$s=1$ in~\eqref{hcp2} instead to obtain
\begin{align*}
L_{3}^{\text{HCP}}(1)  
&=\frac{\pi^2}{8}
+ \pi\;\sqrt\frac{27}{8} \;(\sqrt{3}+1)\;
\zeta\left(\frac12\right)\,L_{-3}\left(\frac12\right)
 \\
&\quad
+ \pi\, \sqrt{\frac{3}{2}}\,
\sum_{N=1}^\infty \frac{u_2(N)}{\sqrt{N}}\,
\left(\frac{1}{e^{\pi \sqrt{3N/2}}-1} \right)
 \\
& \quad -  \frac{3\pi}{\sqrt{8}}\,
\sum_{N=0}^\infty \frac{u_2(3N+1)}{\sqrt{3N+1}} 
\left(\frac{1}{e^{\pi \sqrt{(3N+1)/2}}+1} \right) \\
&= -11.43265\;30014\;95285\;63572\cdots.
\end{align*}
We end this section by noting a connection between two of the values in the above analysis.
By setting~$A=1$ in each of~\eqref{remarkable1} and~\eqref{remarkable2} we obtain the remarkable result
\begin{equation}
{\sum_{i,j,k}}^{\prime} \;\;\;\left. \frac{1}{(i^2+j^2+k^2)^{s}}\right|_{s=1}
=\pi \; {\sum_{i,j,k}}^{\prime} \;\;\;\left. \frac{1}{(i^2+j^2+k^2)^{s}}\right|_{s=1/2}.
\end{equation}
This is consistent with~\cite[p. 46 (1.3.44)]{BorweinEtAl} and is the special case~$s=1$ of the functional equation
\begin{equation}
\label{fe0}
\pi^{-s} \Gamma(s) T_1(1;s) = \pi^{-(\frac32-s)} \Gamma\left(\frac32-s\right) T_1\left(1;\frac32-s\right).
\end{equation}
This functional equation can be deduced from the two formulas for $T_1(A;s)$ in~\eqref{sumpart1} and~\eqref{sumpart2}, as follows.
Replace $s$ with $\frac32-s$ in~\eqref{sumpart1}, then multiply by $\pi^{s-\frac32}\Gamma(\frac32-s)$ and set $A=1$ to get
\begin{align*}
\lefteqn{\pi^{s-\frac32}\Gamma\left(\frac32-s\right)T_1\left(1;\frac32-s\right)} \\
&= 4\pi^{s-\frac32}\Gamma\left(\frac32-s\right)\zeta\left(\frac32-s\right)L_{-4}\left(\frac32-s\right)  +2\pi^{s-\frac12}\Gamma\left(\frac12-s\right) \zeta(1-2s) \nonumber \\
&\quad +4 \sum_{i=1}^\infty\sum_{N=1}^\infty  r_2(N) \left(\frac{N}{i^2}\right)^{(\frac12-s)/2} K_{\frac12-s}\left(2\pi i \sqrt{N}\right),
\end{align*}
where we have used the functional equation for the gamma function in the form
$$
\Gamma(3/2-s) = (1/2-s)\Gamma(1/2-s)
$$
to obtain the second term on the right hand side.
Now apply the functional equations~\eqref{bessel1.5},~\eqref{fe1} and \eqref{fe2} to deduce
\begin{align*}
\lefteqn{\pi^{-(\frac32-s)}\Gamma\left(\frac32-s\right)T_1\left(1;\frac32-s\right)} \\
&= 4\pi^{\frac12-s}\,\Gamma\left(s-\frac12\right)\zeta\left(s-\frac12\right)L_{-4}\left(s-\frac12\right)  +2\pi^{-s}\,\Gamma(s) \zeta(2s) \nonumber \\
&\quad +4 \sum_{i=1}^\infty\sum_{N=1}^\infty  r_2(N) \left(\frac{i}{\sqrt{N}}\right)^{s-\frac12} K_{s-\frac12}\left(2\pi i \sqrt{N}\right).
\end{align*}
The functional equation~\eqref{fe0} follows from this by using~\eqref{sumpart2}.
In addition to providing another proof of the functional equation, the calculation above also demonstrates the interconnection between the
formulas~\eqref{sumpart1} and~\eqref{sumpart2}. Further functional equations of this type are considered in~\cite[p. 46]{BorweinEtAl}.


\subsection{Values at $s=0,\,-1,\,-2,\,-3,\ldots$}
\label{sequals0}
\noindent
Recall from~\eqref{L3formula1} that
\begin{align*}
L(A;s) 
&= 4\left(\frac{A+1}{2}\right)^s \zeta(s)L_{-4}(s) + \frac{\pi A}{s-1} \left(1+\frac{1}{A}\right)^s \zeta(2s-2) \\
&\quad + \frac{2\pi^s\sqrt{A}}{\Gamma(s)} \left(\sqrt{A}+\frac{1}{\sqrt{A}}\right)^s \sum_{i=1}^\infty\sum_{N=1}^\infty  r_2(N) \left(\frac{N}{i^2}\right)^{(s-1)/2} K_{s-1}\left(2\pi i \sqrt{AN}\right) \\
&\quad + \frac{2\pi^s\sqrt{A}}{\Gamma(s)} \left(\sqrt{A}+\frac{1}{\sqrt{A}}\right)^s
\sum_{i=1}^\infty\sum_{N=0}^\infty  (-1)^ir_2(4N+1) \nonumber \\
&\qquad\qquad \times \left(\frac{2N+\frac12}{i^2}\right)^{(s-1)/2} K_{s-1}\left(2\pi i \sqrt{A(2N+\frac12)}\right).
\end{align*}
On using the values $\zeta(0) = -\frac12$, $\zeta(-2) = 0$, $L_{-4}(0) = \frac12$ and the limiting value
$$
\lim_{s\rightarrow 0} \frac{1}{\Gamma(s)} = 0
$$
we readily obtain the result
$$
L(A;0) = -1.
$$
Moreover, since
$$
\zeta(-2)=\zeta(-4)=\zeta(-6) = \cdots = 0,
$$
$$
L_{-4}(-1)=L_{-4}(-3)=\zeta(-5) = \cdots = 0,
$$
and
$$
\lim_{s\rightarrow N} \frac{1}{\Gamma(s)} = 0\quad \text{if $N=0,\, -1,\, -2,\, \cdots$}
$$
it follows that
$$
L(A;-1) = L(A;-2) = L(A;-3) = \cdots = 0.
$$
In a similar way, it can be shown using~\eqref{hcp1} that 
$$L_{3}^{\text{HCP}}(0) = -1
$$
and
$$
L_3^{\text{HCP}}(-1) = L_3^{\text{HCP}}(-2) = L_3^{\text{HCP}}(-3) = \cdots = 0.
$$

\vfill
\pagebreak[4]

\section{Graphs}
\label{graphs}
The formulas ~\eqref{L3formula1}, \eqref{L3formula2}, \eqref{hcp1} and~\eqref{hcp2} have been used to produce the following graphs of $y=L_{3}^{\text{FCC}}(s)$
on the intervals $-10<s<10$ and $-7<s<0$ in Figure~\ref{fig3}. The graph of $y=L_{3}^{\text{HCP}}(s)$ has a similar appearance, and so to allow a comparison the difference
$$
y=L_{3}^{\text{HCP}}(s) - L_{3}^{\text{FCC}}(s)
$$
is plotted using a finer vertical scale in Figure~4.

\begin{figure*}[htbp]
\begin{center}
\includegraphics[scale=0.3, trim=4cm 0 0 0]{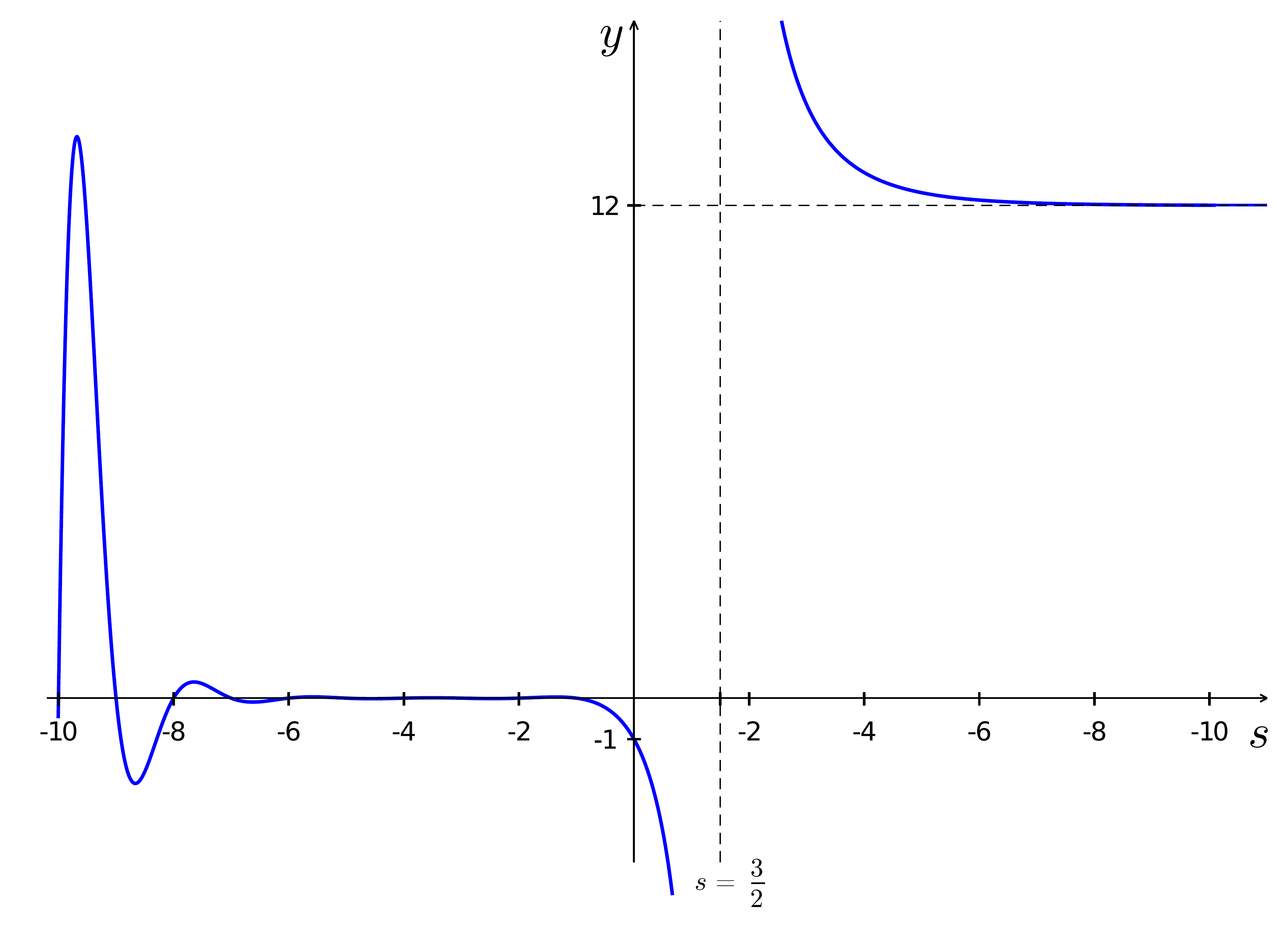}\\
\caption{Graph of $y=L_{3}^{\text{FCC}}(s)$ for $-10<s<10$.}
\label{fig3}
\end{center}
\end{figure*}

\begin{figure*}[htbp]
\begin{center}
\includegraphics[scale=0.3, trim=4cm 0 0 0]{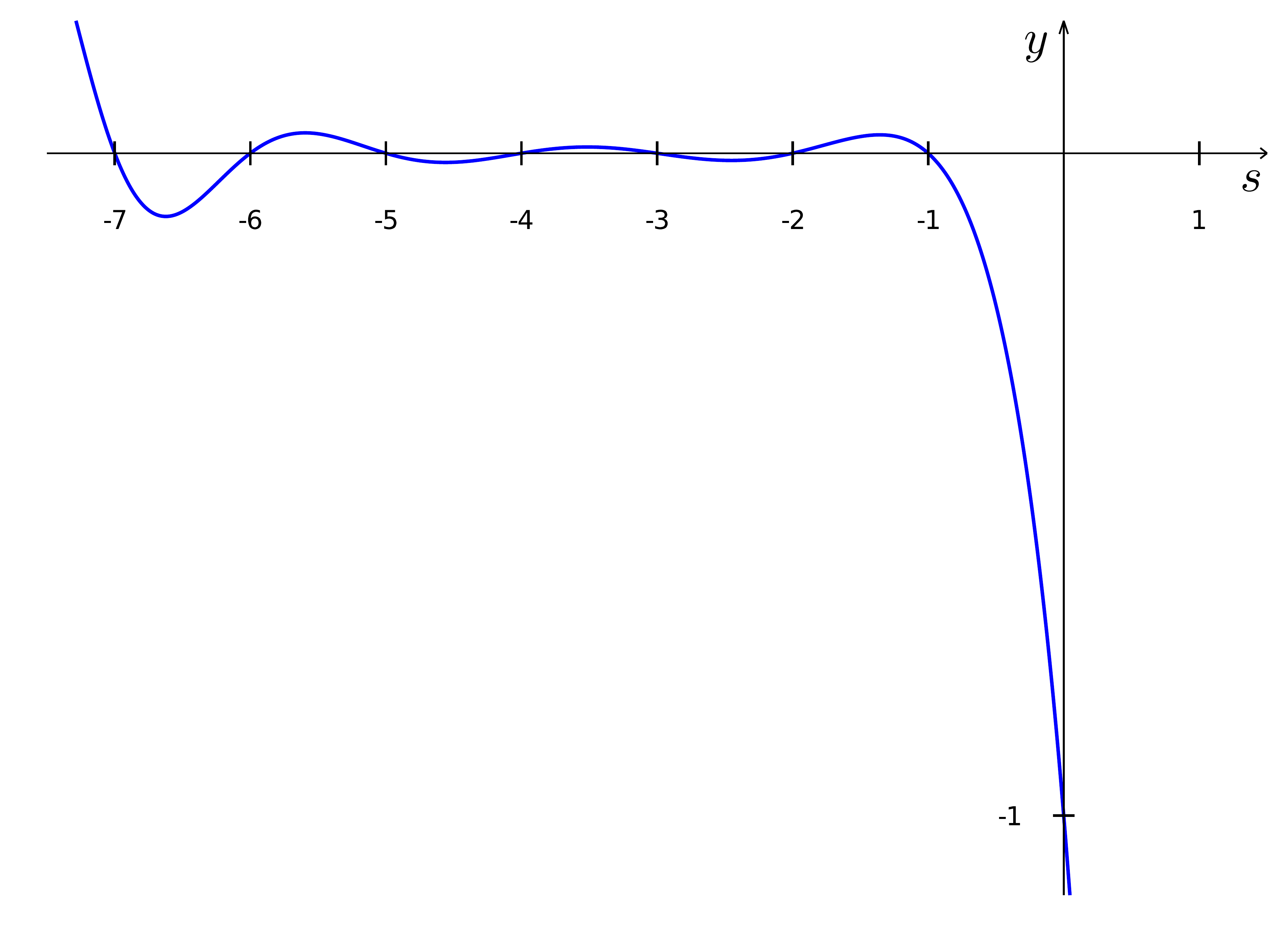}\\
\caption{Graph of $y=L_{3}^{\text{FCC}}(s)$ for $-7<s<0$.}
\label{fig4}
\end{center}
\end{figure*}

\begin{figure*}[htbp]
\begin{center}
\includegraphics[scale=0.3, trim=4cm 0 0 0]{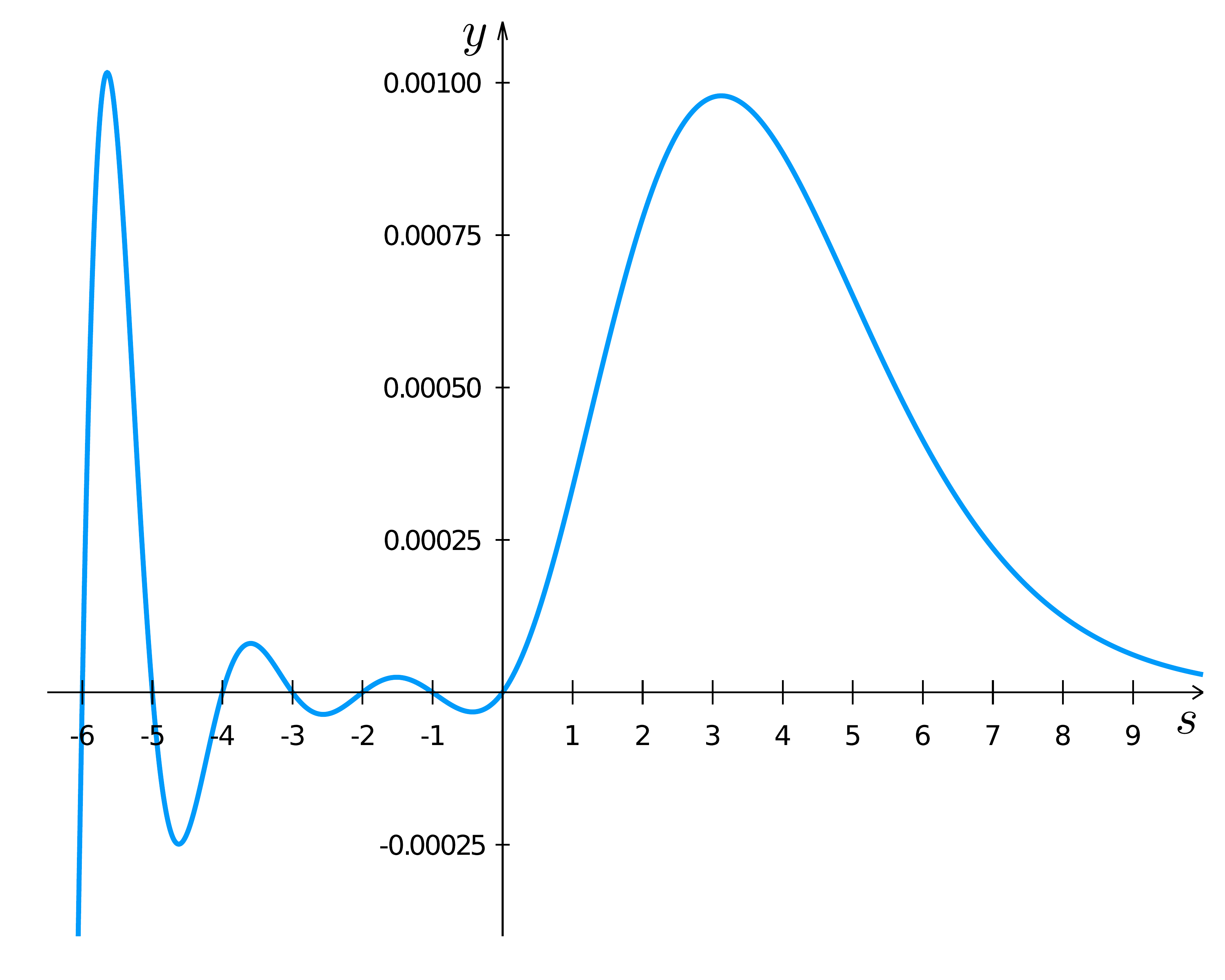}\\
\caption{Graph of $y=L_{3}^{\text{HCP}}(s) - L_{3}^{\text{FCC}}(s)$.}
\label{fig:delta_hcp_fcc}
\end{center}
\end{figure*}

\medskip
\noindent
The graphs appear to suggest the following:

\medskip
\noindent
{\bf Conjecture:}
\begin{align*}
L_{3}^{\text{HCP}}(s) > L_{3}^{\text{FCC}}(s) > 0 &\quad \text{for} \quad s \in \cdots \cup (-6,-5) \cup (-4,-3) \cup (-2,-1) \cup (3/2,\infty)
\\
\\
L_{3}^{\text{HCP}}(s) < L_{3}^{\text{FCC}}(s) < 0 &\quad \text{for} \quad s \in \cdots \cup (-5,-4) \cup (-3,-2) \cup (-1,0)
\end{align*}
and
$$
-1>L_{3}^{\text{HCP}}(s) > L_{3}^{\text{FCC}}(s)  \quad \text{for} \quad s \in (0,3/2).
$$

\appendix

\section{Formulas for special functions}
Many results for special functions and analytic number theory have been used in this work.
For clarity and ease of use, they are stated here along with references.
\subsection{The gamma function}
The gamma function may be defined for $s>0$ by
\begin{equation}
\label{gamma1}
\Gamma(s) = \int_0^\infty  t^{s-1}\, e^{-t}\,\ud t.
\end{equation}
By the change of variable $t=wx$ this can be rewritten in the useful form
\begin{equation}
\label{gamma2}
\frac{1}{w^s} = \frac{1}{\Gamma(s)} \int_0^\infty x^{s-1}\,e^{-wx}\,\ud x.
\end{equation}
See~\mbox{\cite[(1.1.18)]{aar}}.
\subsection{The modified Bessel function}
The following integral may be evaluated in terms of the modified Bessel function:
\begin{equation}
\label{bessel1}
\int_0^\infty x^{s-1}e^{-ax-b/x} \ud x = 2\left(\frac{b}{a}\right)^{s/2} K_s(2\sqrt{ab}).
\end{equation}
By the change of variable $x=u^{-1}$ it can be shown that
\begin{equation}
\label{bessel1.5}
K_{s}(z) = K_{-s}(z).
\end{equation}
When $s=1/2$ the modified Bessel function reduces to an elementary function:
\begin{equation}
\label{bessel2}
K_{1/2}(z) = \sqrt{\frac{\pi}{2z}}\,e^{-z}.
\end{equation}
The asymptotic formula holds:
\begin{equation}
\label{bessel3}
K_s(z) \sim \sqrt{\frac{\pi}{2z}} \, e^{-z} \quad \text{as} \quad z\rightarrow \infty, \quad (\,|\text{arg} \,z| < 3\pi/2).
\end{equation}
For all of these properties, see~\cite[pp. 223, 237]{aar} or~\cite[pp. 233--248]{temme}.

\subsection{Characters}
For an integer $n$, let $\chi_{-4}(n)$ and $\chi_{-3}(n)$ be defined by
\begin{equation}
\label{char1}
\chi_{-4}(n) = \sin(\pi n/2) = \begin{cases}
1 & \text{if $n \equiv 1 \pmod 4$,} \\
-1 & \text{if $n \equiv 3 \pmod 4$,} \\
0 & \text{otherwise}
\end{cases}
\end{equation}
and
\begin{equation}
\label{char2}
\chi_{-3}(n) = \frac{\sin(2\pi n/3)}{\sin(2\pi/3)} = \begin{cases}
1 & \text{if $n \equiv 1 \pmod 3$,} \\
-1 & \text{if $n \equiv 2 \pmod 3$,} \\
0 & \text{otherwise.}
\end{cases}
\end{equation}

\subsection{Theta functions}
The transformation formula for theta functions is~\cite[p. 119]{aar}, \cite[(2.2.5)]{borwein}:
\begin{equation}
\label{theta1}
\sum_{n=-\infty}^\infty e^{-\pi n^2t+2\pi ina} = \frac{1}{\sqrt{t}} \sum_{n=-\infty}^\infty e^{-\pi(n+a)^2/t}, \quad \text{assuming Re$(t)>0$.}
\end{equation}
We will need the special cases $a=0$ and $a=1/2$, which are
\begin{equation}
\label{theta1a}
\sum_{n=-\infty}^\infty e^{-\pi n^2t} = \frac{1}{\sqrt{t}} \sum_{n=-\infty}^\infty e^{-\pi n^2/t}
\end{equation}
and
\begin{equation}
\label{theta1b}
\sum_{n=-\infty}^\infty (-1)^ne^{-\pi n^2t} = \frac{1}{\sqrt{t}} \sum_{n=-\infty}^\infty e^{-\pi(n+\frac12)^2/t}
\end{equation}
respectively.
The sum of two squares formula is~\cite[(3.111)]{cooperbook}
\begin{equation}
\label{theta2}
\left(\sum_{j=-\infty}^\infty q^{j^2}\right)^2 = \sum_{j=-\infty}^\infty \sum_{k=-\infty}^\infty q^{j^2+k^2} = \sum_{N=0}^\infty r_2(N) q^N
\end{equation}
where
\begin{equation}
\label{theta3}
r_2(N) = \#\left\{j^2+k^2=N\right\} = 
\begin{cases} 1 & \text{if $N=0$}, \\ \\
\displaystyle{4\sum_{d|N} \chi_{-4}(d)} &\text{if $N\geq 1$,}
\end{cases}
\end{equation}
the sum being is over the positive divisors $d$ of $N$. For example,
\begin{align*}
r_2(18) &= 4\left(\chi_{-4}(1)+\chi_{-4}(2)+\chi_{-4}(3)+\chi_{-4}(6)+\chi_{-4}(9)+\chi_{-4}(18)\right) \\
&= 4 \left(1+0-1+0+1+0\right) = 4.
\end{align*}
By~\cite[(3.15) and (3.111)]{cooperbook} we also have
\begin{equation}
\label{theta4}
\left(\sum_{j=-\infty}^\infty q^{(j+\frac12)^2}\right)^2 = \sum_{N=0}^\infty r_2(4N+1)q^{(4N+1)/2}.
\end{equation}
\subsection{The cubic theta function}
The cubic analogues of the transformation formula are \cite[(2.2)]{borwein}, \cite[Cor. 5.19]{coopercubic}
\begin{equation}
\label{theta5}
\sum_{j=-\infty}^\infty\sum_{k=-\infty}^\infty e^{-2\pi(j^2+jk+k^2)t} = \frac{1}{\sqrt{3}} \sum_{j=-\infty}^\infty\sum_{k=-\infty}^\infty e^{-2\pi(j^2+jk+k^2)/3t} 
\end{equation}
and
\begin{equation}
\label{theta5a}
\sum_{j=-\infty}^\infty\sum_{k=-\infty}^\infty e^{-2\pi((j+\frac13)^2+(j+\frac13)(k+\frac13)+(k+\frac13)^2)t} 
= \frac{1}{\sqrt{3}} \sum_{j=-\infty}^\infty\sum_{k=-\infty}^\infty \omega^{j-k}e^{-2\pi(j^2+jk+k^2)/3t} 
\end{equation}
where $\omega=\exp(2\pi i/3)$ is a primitive cube root of unity.
The analogue of the sum of two squares result is~\cite[(3.124)]{cooperbook}
\begin{equation}
\label{theta2a}
\sum_{j=-\infty}^\infty \sum_{k=-\infty}^\infty q^{j^2+jk+k^2} = \sum_{N=0}^\infty u_2(N) q^N
\end{equation}
where
\begin{equation}
\label{theta3a}
u_2(N) = \#\left\{j^2+jk+k^2=N\right\} = 
\begin{cases} 1 & \text{if $N=0$}, \\ \\
\displaystyle{6\sum_{d|N} \chi_{-3}(d)} &\text{if $N\geq 1$,}
\end{cases}
\end{equation}
where the sum is again over the positive divisors $d$ of $N$. By~\cite[(3.18) and (3.124)]{cooperbook} we also have
\begin{equation}
\label{theta2b}
\sum_{j=-\infty}^\infty \sum_{k=-\infty}^\infty q^{(j+\frac13)^2+(j+\frac13)(k+\frac13)+(k+\frac13)^2} = \frac12\sum_{N=0}^\infty u_2(3N+1) q^{N+\frac13}
\end{equation}
which is the analogue of~\eqref{theta4}.

\subsection{The Riemann zeta function and $L$ functions}
The definitions are:
\begin{align}
\zeta(s) &= \sum_{j=1}^\infty \frac{1}{j^s} \label{zeta1} \\
L_{-4}(s) &= \sum_{j=1}^\infty \frac{\chi_{-4}(j)}{j^s} = 1-\frac{1}{3^s}+\frac{1}{5^s}-\frac{1}{7^s}+\cdots.   \label{zeta2}   \\
L_{-3}(s) &= \sum_{j=1}^\infty \frac{\chi_{-3}(j)}{j^s}  = 1-\frac{1}{2^s}+\frac{1}{4^s}-\frac{1}{5^s}+\frac1{7^s}-\frac{1}{8^s}+\cdots. \label{zeta3}
\end{align}
The function $\zeta(s)$ is the Riemann zeta function. It has a pole of order~$1$ at $s=1$, and in fact
\begin{equation}
\label{zetapole}
\lim_{s\rightarrow 1} (s-1)\zeta(s) = 1.
\end{equation}
This is a consequence of~\cite[(1.3.2)]{aar}. See also~\cite[p. 58]{temme}.
\\
We will require the functional equations
\begin{equation}
\label{fe1}
\pi^{-s/2} \Gamma (s/2)\zeta(s) = \pi^{-(1-s)/2} \Gamma((1-s)/2)\zeta(1-s)
\end{equation}
and
\begin{equation}
\label{fe2}
\pi^{-s} \Gamma\left(s\right)\zeta(s)L_{-4}(s) = \pi^{-(1-s)} \Gamma\left(1-s\right)\zeta(1-s)L_{-4}(1-s) 
\end{equation}
and the special values
\begin{equation}
\label{zetavalues}
\zeta(2) = \frac{\pi^2}{6}, \quad \zeta(0)= -\frac12,\quad \zeta(-1) = -\frac{1}{12}, \quad \zeta(-2)=\zeta(-4)=\zeta(-6)=\cdots = 0,
\end{equation}
\begin{equation}
\label{L4values}
L_{-4}(1) = \frac{\pi}{4},  \quad L_{-4}(0) = \frac12, \quad L_{-4}(-1) = L_{-4}(-3) = L_{-4}(-5) = \cdots = 0,
\end{equation}
and
\begin{equation}
\label{L3values}
L_{-3}(1) = \frac{\pi \sqrt{3}}{9}, \quad L_{-3}(0) =\frac13, \quad L_{-3}(-1) = L_{-3}(-3) = L_{-3}(-5) = \cdots = 0.
\end{equation}
See~\cite[Ch. 12]{apostol} or~\cite{zucker}.
Other results used are
\begin{equation}
\label{zeta4a}
\sum_{j=0}^\infty \frac{1}{(j+\frac12)^s} = (2^s-1)\zeta(s)
\end{equation}
\begin{equation}
\label{zeta4}
\sum_{j=1}^\infty \frac{(-1)^j}{j^s} = -(1-2^{1-s})\zeta(s)
\end{equation}

\begin{equation}
\label{zeta5}
{\sum_{j,k}}^{\prime} \;\;\; \frac{1}{(j^2+k^2)^s} = 4\zeta(s)L_{-4}(s)
\end{equation}

\begin{equation}
\label{zeta6}
{\sum_{j,k}}^{\prime} \;\;\; \frac{(-1)^{j+k}}{(j^2+k^2)^s}= -4(1-2^{1-s})\zeta(s)L_{-4}(s).
\end{equation}

\begin{equation}
\label{zeta7}
{\sum_{i,j}}^{\;\prime} \frac{1}{(i^2+ij+j^2)^{s}}
= 6\zeta(s) L_{-3}(s)
\end{equation}

\begin{equation}
\label{zeta8}
{\sum_{i,j}} \frac{1}{((i+\frac13)^2+(i+\frac13)(j+\frac13)+(j+\frac13)^2)^{s}}
= 3(3^s-1)\zeta(s) L_{-3}(s).
\end{equation}
The identities~\eqref{zeta4a} and~\eqref{zeta4} follow from the definition of~$\zeta(s)$ by series rearrangements.
For~\eqref{zeta5}, \eqref{zeta6} and \eqref{zeta7}, see (1.4.14), (1.7.5) and (1.4.16), respectively, of~\cite{BorweinEtAl}.
The identity~\eqref{zeta8} can be obtained by the method of Mellin transforms (e.g., see Appendix~A of~\cite{paper1}) starting with~\cite[(3.36)]{cooperbook}.

\section{Behaviour as $A\rightarrow 0^+$ and $A\rightarrow+\infty$}
We briefly consider the behaviour of the lattices in the limiting cases $A\rightarrow 0^+$ and $A\rightarrow+\infty$.
Some of the basis vectors become infinite in the limit, leaving a sublattice of lower dimension. We discuss each case
$A\rightarrow 0^+$ and $A\rightarrow+\infty$ both in terms of theta functions
and then in terms of the basis vectors.

First, consider the limit $A\rightarrow 0^+$. In the interval $0<A<1/3$ the theta function is
\begin{align*}
\theta(A;q) &= \sum_{i=-\infty}^\infty \sum_{j=-\infty}^\infty \sum_{k=-\infty}^\infty q^{g(A;i,j,k)} \\
&= \sum_{i=-\infty}^\infty \sum_{j=-\infty}^\infty \sum_{k=-\infty}^\infty q^{(A(i+j)^2+(j+k)^2+(i+k)^2)/4A}.
\end{align*}
As $A\rightarrow 0^+$ we have
$$
q^{(j+k)^2/4A} \rightarrow 0\quad\text{and}\quad q^{(i+k)^2/4A} \rightarrow 0
$$
unless $j=-k$ and $i=-k$, respectively. Hence,
\begin{align*}
\lim_{A\rightarrow 0^+}\theta(A;q)
&= \lim_{A\rightarrow 0^+} \sum_{k=-\infty}^\infty \left(\sum_{i=-k} \sum_{j=-k}q^{(A(i+j)^2+(j+k)^2+(i+k)^2)/4A} \right) \\
&= \lim_{A\rightarrow 0^+} \sum_{k=-\infty}^\infty q^{A(-k-k)^2/4A} \\
&=\sum_{k=-\infty}^\infty q^{k^2}.
 \end{align*}
This corresponds to the one-dimensional lattice with minimal distance~$1$. The kissing number is~2, which is in agreement with the other lattices
in the range \mbox{$0<A<1/3$}, as indicated in Table~1.
In terms of the basis vectors, from~\eqref{b1b2b3} we have
$$
\vect{b}_1 = \left(\frac12,\frac{1}{2\sqrt{A}},0\right)^\top, \quad
\vect{b}_2 = \left(\frac12,0,\frac{1}{2\sqrt{A}}\right)^\top, \quad
\vect{b}_3 = \left(0,\frac{1}{2\sqrt{A}},\frac{1}{2\sqrt{A}}\right)^\top.
$$
The only linear combinations $\vect{v} = i \vect{b}_1 + j \vect{b}_2 + k \vect{b}_3$ (for $i,j,k\in \mathbf{Z}$) that remain finite in the limit~\mbox{$A\rightarrow 0^+$} occur when $i=-k$, $j=-k$ in which case we
obtain
$$
\vect{v} = -k \vect{b}_1 -k \vect{b}_2 + k \vect{b}_3 = -k(1,0,0)^\top.
$$
That is, the limiting lattice is just the one-dimensional lattice consisting of integer multiples of $(1,0,0)^\top$.

Now consider the limit $A\rightarrow +\infty$. For $A>1$ the theta function is 
\begin{align*}
\theta(A;q) &= \sum_{i=-\infty}^\infty \sum_{j=-\infty}^\infty \sum_{k=-\infty}^\infty q^{g(A;i,j,k)} \\
&= \sum_{i=-\infty}^\infty \sum_{j=-\infty}^\infty \sum_{k=-\infty}^\infty q^{(A(i+j)^2+(j+k)^2+(i+k)^2)/2}.
\end{align*}
Since $q^{A(i+j)^2/2} \rightarrow 0$ as $A\rightarrow +\infty$ unless $i=-j$, it follows that
\begin{align*}
\lim_{A\rightarrow +\infty}\theta(A;q)
&=\sum_{j=-\infty}^\infty \sum_{k=-\infty}^\infty \left( \sum_{i=-j}  q^{(A(i+j)^2+(j+k)^2+(i+k)^2)/2}\right) \\
&=\sum_{j=-\infty}^\infty \sum_{k=-\infty}^\infty  q^{((j+k)^2+(-j+k)^2)/2} \\
&= \sum_{j=-\infty}^\infty \sum_{k=-\infty}^\infty q^{j^2+k^2}.
\end{align*}
This is the theta series for the two-dimensional square close packing lattice with minimal distance~$1$. The kissing number is~$4$, in agreement
with other values in the range~$A>1$ given by Table 1.
In terms of the basis vectors, from~\eqref{b1b2b3} we have
$$
\vect{b}_1 = \frac{1}{\sqrt{2}}(\sqrt{A},1,0)^\top, \quad
\vect{b}_2 = \frac{1}{\sqrt{2}}(\sqrt{A},0,1)^\top, \quad
\vect{b}_3 = \frac{1}{\sqrt{2}}(0,1,1)^\top.
$$
The only linear combinations $\vect{v} = i \vect{b}_1 + j \vect{b}_2 + k \vect{b}_3$ (for $i,j,k\in \mathbf{Z}$) that remain finite in the limit~\mbox{$A\rightarrow +\infty$}
occur when $i=-j$, in which case we obtain
$$
\vect{v} = -j \vect{b}_1 +j \vect{b}_2 + k \vect{b}_3 = \frac{1}{\sqrt{2}} \left[ j(0,-1,1)^\top + k(0,1,1)^\top\right].
$$
This is isomorphic to the two-dimensional square close packing lattice with minimal distance~$1$, rotated from the coordinate axes by 45 degrees.


\begin{thebibliography}{99}

\bibitem{aar}
G. Andrews, R. Askey and R. Roy,
{\em Special Functions,}
Cambridge University Press, Cambridge, 1999.

\bibitem{apostol}
T. M. Apostol,
{\em Introduction to Analytic Number Theory,}
Springer, New York, 1976.

\bibitem{borwein}
J. M. Borwein and P. B. Borwein,
{\em A cubic counterpart of Jacobi's identity and the AGM,}
Trans. Amer. Math. Soc. {\bf 323} (2) 691--701 (1991).

\bibitem{borwein1998convergence}
D. Borwein, J. M. Borwein, and C. Pinner,
{\em Convergence of Madelung-like lattice sums},
Trans. Am. Math. Soc. {\bf 350}, 3131--3167 (1998).

\bibitem{BorweinEtAl}
J. M. Borwein, M. L. Glasser, R. C. McPhedran, J. G. Wan, and I. J. Zucker,
{\em Lattice Sums Then and Now,}
Cambridge University Press, Cambridge, 2013.

\bibitem{paper1}
A. Burrows, S. Cooper, E. Pahl and P. Schwerdtfeger,
{\em{Analytical methods for fast converging lattice
sums for cubic and hexagonal close-packed
structures,}}
J. Math. Phys. {\bf 61}, 123503 (2020),
\url{https://doi.org/10.1063/5.0021159}

\bibitem{duals}
J. H. Conway and N. J. A. Sloane,
{\em On lattices equivalent to their duals,}
J. Number Theory {\bf 48}, 373--382 (1994).

\bibitem{conway2007optimal}
J. H. Conway and N. J. A. Sloane,
{\em The optimal isodual lattice quantizer in three dimensions},
Adv. Math. Commun. {\bf 1} (2007),  257--260.
arXiv preprint math/0701080 (2007).

\bibitem{ConwayBook}
J. H. Conway and N. J. A. Sloane, 
{\em Sphere packings, lattices and groups,} 
Vol. 290. Springer Science \& Business Media, 2013.

\bibitem{coopercubic}
S. Cooper,
{\em Cubic theta functions,}
J. Comput. Appl. Math. {\bf 160}, 77--94 (2003).

\bibitem{cooperbook}
S. Cooper,
{\em Ramanujan's Theta Functions,}
Springer, New York, 2017.

\bibitem{furth1944}
R. F{\"u}rth,
{\em On the equation of state for solids,}
Proc. Royal Soc. Lond. A. Math Phys. Sci. {\bf 183}, 87--110 (1944).

\bibitem{Gruneisen1912}
E. Gr\"uneisen,
{\em Theorie des festen Zustandes einatomiger Elemente},
Ann. Phys. {\bf 344}, 257--306 (1912).

\bibitem{Jones-1925}
J. E. Jones and A. E. Ingham,
{\em On the Calculation of Certain Crystal Potential Constants, and on the Cubic Crystal of Least Potential Energy,}
Proc. Royal Soc. Lond. A. Math Phys. Sci. {\bf107}, 636--653 (1925).

\bibitem{Kane}
G. Kane, and M. Goeppert-Mayer,
{\em Lattice Summations for Hexagonal Close-Packed Crystals,}
J. Chem. Phys. {\bf 8}, 642 (1940).

\bibitem{levinson}
N. Levinson and R. M. Redheffer,
{\em Complex Variables,}
Holden-Day, San Francisco, 1970.

\bibitem{madelung1918}
E. Madelung,
{\em Das elektrische {F}eld in {S}ystemen von regelm{\"a}{\ss}ig angeordneten {P}unktladungen,}
Phys. Z. {\bf 19}, 32 (1918).

\bibitem{Smits2021}
P. Schwerdtfeger, A. Burrows and O. R. Smits, 
{\em The Lennard Jones Potential Revisited--Analytical Expressions for Vibrational Effects in Cubic and Hexagonal Close-Packed Lattices,}
J. Phys. Chem.  A {\bf 125}, 3037--3057 (2021).

\bibitem{selberg}
A. Selberg and S. Chowla,
{\em On Epstein's zeta-function,}
J. Reine Angew. Math. {\bf 227}, 86--110 (1967).

\bibitem{Stillinger2001}
F. H. Stillinger, 
{\em Lattice sums and their phase diagram implications for the classical Lennard-Jones model,}
J. Chem. Phys. {\bf 115}, 5208--5212 (2001).

\bibitem{temme}
N. Temme,
{\em Special Functions. An Introduction to the Classical Functions of Mathematical Physics,}
Wiley, New York, 1996.

\bibitem{zucker}
I. J. Zucker and M. M. Robertson,
{\em Exact values for some two-dimensional lattice sums,}
J. Phys. A {\bf 8}, 874--881, (1975). 

\end{thebibliography}
\end{document}